\documentclass[journal,transmag]{IEEEtran}
\usepackage{amsmath,amsthm,amssymb,amscd,graphicx,latexsym,amsfonts,mathrsfs}
\usepackage{graphicx}
\usepackage{color}

\usepackage{enumerate}
\usepackage{dsfont}
\usepackage{chemarrow,stackrel}
\usepackage[numbers,sort&compress]{natbib}

\newtheorem{theorem}{Theorem}
\newtheorem{lemma}{Lemma}
\newtheorem{definition}{Definition}
\newtheorem{assumption}{Assumption}
\newtheorem{remark}{Remark}

\def\EE{{\mathbb E}}
\def\NN{{\mathbb N}}
\def\PP{{\mathbb P}}
\def\RR{{\mathbb R}}

\def\d{{\mathrm d}}
\def\dd{{\mathscr D}}
\def\ff{{\mathcal F}}
\def\lf{{\mathcal L}}

\def\lagr{{\mathcal L}}  
\def\ll{{\mathscr L}}

\def\tll{{\tilde{\mathscr L}}}
\def\mm{ { \mathcal M} }

\def\al{\alpha}
\def\bt{\beta}

\def\gm{\gamma}

\def\lmd{\lambda}

\def\bA{\bar{A}}

\def\bal{\bar{\alpha}}

\def\bB{\bar{B}}

\def\bc{\bar{c}}

\def\barf{\bar{F}}
\def\barf{\bar{f}}
\def\barg{\bar{g}}
\def\bH{\bar{H}}

\def\bL{\bar{L}}

\def\bn{\bar{n}}

\def\bQ{ \bar{Q} }
\def\barr{\bar{r}}

\def\bart{\bar{t}}
\def\btau{\bar{\tau}}

\def\bart{\bar{t}}
\def\bu{\bar{u}}

\def\bY{\bar{Y}}

\def\hB{\hat{B}}

\def\hK{\hat{K}}

\def\hm{\hat{m}}
\def\hatq{\widehat{q}}

\def\htau{\widehat{\tau}}

\def\oW{\overline{W}}
\def\oDt{\overline{\Delta t}}

\def\uDt{\underline{\Delta t}}
\def\tWa{\tilde{W}^{\mu}}
\def\btauK{\overline{\tau}_K}   

\def\uDt{\underline{\Delta t}}

\def\tA{\tilde{A}}
\def\tal{\tilde{\alpha}}

\def\tbt{\tilde{\beta}}

\def\tc{\tilde{c}}

\def\tP{\tilde{P}}

\def\tQ{\tilde{Q}}

\def\tT{\widetilde{T}}
\def\ttau{\widetilde{\tau}}

\def\tU{\tilde{U}}

\def\tV{\tilde{V}}
\def\tv{\widetilde{v}}

\def\tW{\tilde{W}}

\begin{document}
%
\title{Stabilization  of cyber-physical  systems: a foundational theory of  computer-mediated
control systems}

\vspace{-1cm}

\author{
\\
\\
\\
\IEEEauthorblockN{Lirong Huang}
\IEEEauthorblockA{Guangzhou 510320, Guangdong, China}
\thanks{Manuscript received Month Day, 2022. 
Corresponding author: L. Huang (email:lrhuang@aliyun.com).}}

%



\IEEEtitleabstractindextext{%
\begin{abstract}

This paper presents the cyber-physcial model of a computer-mediated control system that is a seamless, fully synergistic
integration of the physical system and the cyber system, which provides a systematic framework for synthesis of cyber-physical systems (CPSs). 
In our proposed framework, we establish a Lyapunov stabilty theory for   synthesis of CPSs and apply it to sampled-data control systems, which  are typically  synonymous with computer-mediated control systems. By our CPS approach, we not only develop stability criteria for sampled-data control systems but also reveal the equivalence and inherent relationship between the  two main design methods (viz. controller emulation and discrete-time approximation) in the literature. As application of our established theory, we study feedback stabilization of linear sampled-data stochastic systems and propose a control design method. Illustrative examples show that our proposed method has improved the existing results. Our established   theory   of synthetic CPSs   lays  a theoretic foundation for computer-mediated control systems and provokes  many open and interesting problems for future work.
\end{abstract}

}

\maketitle

\IEEEdisplaynontitleabstractindextext

%
\IEEEpeerreviewmaketitle

\section{Introduction}        \label{sec;introduction}

Feedback mechanisms were discovered and exploited at all levels in nature, which are crucial to homeostasis and life \cite{astrom2014,wiener1961}. As a technology, feedback control can be found in many examples from ancient times. In the modern era, it was fundamental to the industrial evolution that James Watt successfully adapted the centrifugal governor for  the steam engine and, in the later designs, the governor
became an integral part of all steam engines.  Theorectic research on the mechanical systems of governors started with the classical paper of Maxwell that placed stability at the core of his analysis of feedback mechanisms \cite{maxwell1868}. 
Stability analysis and feedback stabilization of dynamical systems are at the core of systems and control theory \cite{astrom1997,astrom2014,boyd1994,huang2010PhD,khalil2002,khas2012,liu1988,liu1992,mao2007book,stein2003,yang2001}. As is well known, the Lyapunov method is an efficient and powerful tool for stability analysis and synthesis of dynamical systems. The investigation of Lyapunov method has been so extensive and intensive  that the Lyapunov-based results can be found in an enormous literature. Lyapunov-type theorems have been developed for stability analysis and   application to feedback stabilization of myriad  systems such as  discrete-time systems \cite{huang2015},   large-scale systems \cite{liu1992},   time-delay systems \cite{fridman2008}, stochastic systems \cite{huang2009a} and a variety of stochastic hybrid systems \cite{teel2014}. As a matter of fact,  Lyapunov-type stability theory finds an extremely wide range of applications including those in numerical analysis \cite{huang_partI} and system identification \cite{huang2012}.

Practically all control systems that are implemented today are based on computer control, which contain both continuous-time signals and sampled, or discrete-time, signals. Such systems have traditionally been called sampled-data systems and have motivated the study of sampled-data control systems \cite{astrom1997,nesic2001}.  
There is a wealth of  impressive results on sampled-data control systems along two main approaches, see, e.g.,  \cite{astrom1997,fridman2004,fridman2010,nagh2008,nesic2001,nesic2004,nesic2006,nghiem2012,oishil2010,seuret2012} and the references therein. The first starts with a designed continuous controller and focuses on discretizing the controller on a sampler and zero-order-hold (ZOH) device, which employs the strategy of  controller emulation and is called the process-oriented view.  The second disccretizes a continuous plant given implementation-dependent sampling times and designs a controller  for the discretized plant, which utilizes some approximate discrete-time model  for controller design and is called the computer-oriented view.  There is  another approach based on the hybrid/impulsive modelling of sampled-data systems which considers the sampled state a pure jump process, see Remark \ref{remark-hybrid} below as well as   \cite{fridman2004,nagh2008,rios2016}. 
Over the recent years, sampled-data control of stochastic systems has also been studied \cite{mao2013,mao2014,you2015} since stochastic modelling has come to play an important role in engineering and science \cite{huang2010PhD,huang2016,mao2007book,scotton2013,teel2014}. 

A new and general class of stochastic impulsive differential equations (SiDEs) is formulated to serve as a canonic form of cyber-physical systems (CPSs) and a foundational theory of   the CPSs is constructed  in \cite{huang_partI}. 
The canonic form of CPSs is composed of  physical and cyber subsystems and it is distinct from the  impulsive systems in the literature \cite{huang202x,samoilenko1995,teel2014,yang2001}, which has been highlighted in \cite{huang_partI}.  In this paper, we study feedback stabilization of the CPSs, that is, synthesis of CPSs for stability of the controlled CPSs.  the results  in \cite{huang_partI} do not apply to such synthesized systems. For this purpose, we construct a general class of SiDEs  for  synthesis of CPSs so that  the states of the physical and the cyber subsystems can both be utilized in a feedback mechanism to control the underlying physical processes.  As a theoretic foundation, we  develop a Lyapunov stability theory for  the synthetic CPSs. Our proposed CPS  theory   has a  very wide range of applications including sampled-data control systems.   Sampled-data control systems  have   an  exemplary structure of CPSs  \cite[Figure 1]{lee2010} and can typically be expressed in our canonic form of synthetic  CPSs. Applying the Lyapunov stability theory, we study stability of sampled-data control systems and address the key questions in the two main approaches, respectively.  By our CPS approach, we not only develop  stability criteria for sampled-data control systems but also disclose the equivalence and intrinsic relationship between the two main design methods in the literature. As application of our established theory, we study feedback stabilization of linear sampled-data stochastic systems and present a control design method. Illustrative examples are given to verify that our proposed method has improved the exsting results significantly. Our proposed canonic form and theory of synthetic CPSs  construct a foundational theory of computer-mediated control systems. In this paper, we initiate a system science for CPSs that arouses many interesting and challenging problems of computer-mediated control systems.

\section{A general class of SiDEs for synthesis of CPSs}  \label{sec:generalSiDEs}

This   paper, unless otherwise specified,   employs the
following notation.  Denote by  $( \Omega , \mathcal{F}, \{ \mathcal{F}_t
\}_{t \ge 0}, \PP ) $ a complete probability space with a
filtration $ \{ \mathcal{F}_t \}_{t \ge 0} $ satisfying the usual
conditions \cite{mao2007book}
and by ${\mathbb{E}}[\cdot]$   the expectation
operator with respect to the probability measure. Let $B(t)= \begin{bmatrix}
B_1(t) & \cdots & B_m(t) \end{bmatrix}^T$ be an $m$-dimensional Brownian motion
defined on the probability space. If $x, y$ are real numbers, then
$x \vee y$ (resp.  $x \wedge y$)  denotes the maximum (resp.   minimum) of $x$ and $y$.  Denote by $A^T$ the transpose of  a vector or a matrix $A$. If $P$ is a square matrix, $P>0$ (resp. $P<0$) means that P is a symmetric positive (resp. negative) definite matrix of appropriate dimensions 
while $ P \ge 0$ (resp. $P \le 0$) is a symmetric positive (resp. negative) semidefinite matrix. Let  $\lambda_M( \cdot) $ and $\lambda_m( \cdot)$ be  a matrix's  eigenvalues   with the
maximum  and   the minimum real parts, respectively, and $| \cdot |$  the
Euclidean norm of a vector and the trace (or Frobenius) norm of a matrix. 
Denote by $I_n$ the $n \times n$ identity matrix and   
 by $0_{n \times m}$ the $n \times m$ the zero matrix, or, simply, by $0$  the zero matrix of appropriate  dimensions.  
Let   $C^{2, 1} (\RR^n \times \RR_+; \RR_+)$ be the family of   all nonnegative  functions $V(x, t)$ on $\RR^n \times \RR_+$ that are continuously twice differentiable in $x$ and once in $t$, and $C^2 (\RR^n ; \RR_+)$  the special class of $C^{2, 1} (\RR^n \times \RR_+; \RR_+)$ that is independent of $t$.  
Denote by $C  ([a,b); \RR^n)$ the space of all right continuous $\RR^n$-valued functions $\varphi$ defined on $[a, b)$ with a norm $|| \varphi || = \sup_{a \le \theta < b} | \varphi (\theta) |< \infty$ ,  by $\lf_{\ff_t}^p ( [a, b); \RR^n)$ with $p>0$ the family of all $\ff_t$-measurable $C  ([a,b); \RR^n)$-valued random variables $\varphi $ such that $\sup_{a \le t < b} \EE | \varphi (t) |^p < \infty$   
and by $\mm^p ([a, b]; \RR^n)$   the family of $\RR^n$-valued adapted process $\{ \varphi (t): a \le t \le b \}$ such that $\EE \int_a^b |\varphi (t)|^p \d t < \infty$. 
 Let $\NN  $ be the set of all natural numers and $\Xi^m_\NN $ be the set of all independent and identically distributed sequences $\{ \xi (k) \}_{k \in \NN} $   with $\xi (k)= \begin{bmatrix}
\xi_1(k) & \cdots & \xi_m(k) \end{bmatrix}^T$ and $ \xi_j(k) $ obeying standard Gaussian distribution for $j =1, 2, \cdots, m$.  Sequence 
$\{t_k\}_{k \in \NN}$ with $ t_1 > t_0 :=0$ is  strictly increasing   and satisfies $0<\underline{\Delta t} :=\inf_{k \in \NN} \{ t_k - t_{k-1} \} \le \overline{\Delta t} :=\sup_{k \in \NN} \{ t_k -t_{k-1} \} < \infty$  and hence $t_k \to \infty$ as $k \to \infty$.  Let   $t_\ast = \sup \{ t_k: t \ge t_k, k \ge 0 \}$ for all  $ t \ge 0$ and   $\varphi_t = \{ \varphi (\theta): t_\ast   \le \theta \le t \}$  for all $ \varphi \in C  ( [ t_{k-1}, t_k ); \RR^n)$ and $t \in [t_{k-1}, t_k)$. 

Let us consider  the following
  stochastic impulsive system described by SiDEs  
\begin{subequations}  \label{SiDE-xy}
\begin{align}
& \mathrm{d}  x(t) = f(x(t), y(t), t)\mathrm{d}t + g(x(t), y(t), t )\mathrm{d} B(t)   \label{SDE_x} \\
& \hspace{5cm} \;\; {}  t \in [0, \infty) \nonumber   \\
& \d y(t)  = \tilde{f} (x(t), y(t), t) \d t + \tilde{g}( x(t), y(t), t) \d B (t)   \label{SDE_y}  \\
& \hspace{4.8cm} {} t \in [ 0, \infty) \setminus \{t_k\}_{k \in \NN}  \nonumber   \\
&\tilde{\Delta} (x_{t_k^-}, y_{t_k^-},  k ) := y(t_k) - y(t_k^-)  \nonumber  \\
 & \; \; {}       = \tilde{h}_f (x_{t_k^-}, y_{t_k^-}, k)     +  \bar{h}_g( x_{t_k^-}, y_{t_k^-},    k ) \bar{ \xi } (k )   \quad  k \in \NN    \label{Impulse_y} 
\end{align}
\end{subequations}
 with initial values   $x(0) =x_0 \in \RR^n $ 
and $ y(0)  = y_0 \in \RR^q$, 
where  measurement noise  $\bar{\xi} \in \Xi^n_\NN $ with $ \bar{\xi} (k)$ being independent of $\{ x(t), y(t), B(t):  0 \le t < t_k  \}$ for all $k \in \NN$;  $ f : \RR^n \times \RR^q \times \RR_+ \to \RR^n $, $g : \RR^n \times \RR^q \times \RR_+ \to \RR^{n \times m} $,  
$\tilde{f} : \RR^n \times \RR^q \times \RR_+ \to \RR^q $, $\tilde{g} : \RR^n \times \RR^q \times \RR_+ \to \RR^{q \times q} $, 
$\tilde{h}_f : C  ([t_{k-1}, t_k); \RR^n) \times C  ([t_{k-1}, t_k); \RR^q)  \times \NN \to \RR^q $  and $\bar{h}_g :  C  ([t_{k-1}, t_k); \RR^n) \times C  ([t_{k-1}, t_k); \RR^q)  \times \NN \to \RR^{q \times n} $ are measurable functions that obey $f (0, 0, t) =0 ,    g(0, 0,  t)=0,    
\tilde{f} (0, 0, t) =0,      \tilde{g}(0,0, t)=0,    
 \tilde{h}_f(0,0, k) =0,          \bar{h}_g (0,0, k) =0  $ 
for all $t \in \RR_+$ and $k \in \NN$ and they satisfy the local Lipschitz  condition and the linear growth condition  specified as Assumption \ref{localLipschitz} and Assumption \ref{lineargrowth}, respectively.
\begin{assumption}   \label{localLipschitz}
  For every integer $\bn \ge 1$, there is a constant $ L_{\bn}> 0$  such that 
\begin{multline}   
    |f(x, y, t) - f(\bar{x}, \bar{y}, t) |^2  \vee |g(x, y, t)- g(\bar{x}, \bar{y}, t)|^2        \\
 \vee  | \tilde{f}(x, y, t) - \tilde{f}(\bar{x}, \bar{y}, t) |^2 \vee | \tilde{g}(x, y, t) - \tilde{g}(\bar{x}, \bar{y}, t) |^2         \\          
           \le L_{\bn}  ( |x- \bar{x}| \vee | y - \bar{y} | )^2  \label{localLipschitz-x}
\end{multline}
for all $(x, y, \bar{x}, \bar{y} ) \in \RR^n \times \RR^q \times \RR^n \times \RR^q$ with $|x| \vee |y| \vee |\bar{x} | \vee | \bar{y}| \le \bn$ and  $t \in \RR_+$;   and  there is   a constant $\tilde{L}_{\bn} >0$ such that
\begin{multline} 
  | \tilde{h}_f(x_{t_k^-}, y_{t_k^-}, k) - \tilde{h}_f(\tilde{x}_{t_k^-}, \tilde{y}_{t_k^-}, k) |^2          \\
   {}     \vee | \bar{h}_g (x_{t_k^-}, y_{t_k^-}, k) - \bar{h}_g ( \tilde{x}_{t_k^-}, \tilde{y}_{t_k^-}, k) |^2       \\
   \le \tilde{L}_{\bn} (   || x_{t_k^-} - \tilde{x}_{t_k^-} || \vee  || y_{t_k^-} - \tilde{y}_{t_k^-} ||  )^2  \label{localLipschitz-y}
\end{multline}
for all those  $(x_{t_k^-}, y_{t_k^-}, \tilde{x}_{t_k^-}, \tilde{y}_{t_k^-} )  \in C  ([t_{k-1}, t_k); \RR^n) \times C  ([t_{k-1}, t_k); \RR^q)   \times C  ([t_{k-1}, t_k); \RR^n) \times C  ([t_{k-1}, t_k); \RR^q) $ with $|| x_{t_k^-}|| \vee || y_{t_k^-}|| \vee || \tilde{x}_{t_k^-}|| \vee || \tilde{y}_{t_k^-} || \le \bn$  and $k\in \NN$.
\end{assumption}
\begin{assumption}  \label{lineargrowth}
  There is a  constant $L>0$   such that 
\begin{multline}   
   |f(x, y, t)  |^2 \vee |g(x, y, t)|^2 \vee  | \tilde{f}(x, y, t)   |^2 \vee | \tilde{g}(x, y, t)   |^2 \\
        \le L  ( |x|   \vee  | y   | )^2   \label{lineargrowth-x}
\end{multline}
for all $(x, y ) \in \RR^n \times \RR^q$ and $t \in \RR_+$;  and  there is a   constant  $\tilde{L}  >0$   such that
\begin{multline} 
      | \tilde{h}_f(x_{t_k^-}, y_{t_k^-}, k)   |^2      
        \vee | \bar{h}_g (x_{t_k^-}, y_{t_k^-}, k)   |^2        \\
    \le \tilde{L}   (   || x_{t_k^-}  ||  \vee  || y_{t_k^-}  ||  )^2 
 \label{lineargrowth-y}
\end{multline}
for  all  $(x_{t_k^-}, y_{t_k^-})  \in C  ([t_{k-1}, t_k); \RR^n) \times C  ([t_{k-1}, t_k); \RR^q)$  and $k\in \NN$.
\end{assumption}

SiDE (\ref{SiDE-xy}) is construced to serves as the canonic form  for synthesis of CPSs  in which both $x(t)$ and $y(t)$ can be utilized in some feedback mechanism to steer the   physical subsystem. 
Actually,  CPS \cite[Eq.(2.1)]{huang_partI}   is a particular case of SiDE (\ref{SiDE-xy})   
as the impluses on    subsystem $x(t)$ and the simulation sequence are omitted for the sake of simplicity. 
The canonic form  (\ref{SiDE-xy})  of synthetic CPSs   exploits our knowledge of both the physical and the cyber sides to control the underlying physical  processes. It has a  wide range of applications, which, for example,  can represent the CPS dynamics for not only  feedback stabilization of sampled-data systems but also observer-based control of dynamical systems with impulse effects  such as a robot model in \cite{grizzle2007}. The former is studied  in this paper and the latter  among future work.

Clearly, the trivial  solution is an equilibrium of  system (\ref{SiDE-xy}).  
For a function $V \in C^{2,1} (\RR^n \times \RR^q \times \RR_+; \RR_+)$,  the infinitesimal generator $\ll V: \RR^n \times \RR^q  \times \RR_+ \to \RR$ associated with system (\ref{SDE_x}) is defined as
\begin{multline}     
  \ll V (x, y, t) = V_t (x,  t) + V_x (x,  t) f(x, y, t)     \\
 + \frac{1}{2}   {\rm trace} \left[ g^T (x, y, t) V_{xx} (x,  t) g(x, y, t) \right], 
      \label{LV}
\end{multline}
where 
$ V_t (x,  t) =  \frac{ \partial V(x,  t)} { \partial t}$,   
    $ V_{xx} (x, t) =  \begin{bmatrix}  \frac{ \partial^2 V(x,  t)}{ \partial x_i \, \partial x_j}   \end{bmatrix}_{n \times n}$, 
 $V_x (x,  t) = \begin{bmatrix}  \frac{ \partial V(x,  t)}{ \partial x_1} & \cdots & \frac{ \partial V(x,   t)}{ \partial x_n}   \end{bmatrix}$.
%
Similarly, for a function $\tilde{V} \in C^{2,1} (\RR^q \times \RR_+; \RR_+)$, one can define generator  $\tll \tilde{V}:  \RR^n \times \RR^q \times \RR_+ \to \RR$ associated with system (\ref{SDE_y}) as
\begin{multline}    
  \tll \tilde{V} (x, y, t) = \tilde{V}_t (y, t) + \tilde{V}_y (y, t) \tilde{f}(x, y, t)    \\
 + \frac{1}{2} {\rm trace} \left[ \tilde{g}^T (x, y,  t) \tilde{V}_{yy} (y, t) \tilde{g}(x, y, t) \right]. 
     \label{tLtV}
\end{multline}

Let $z(t) = [ x^T(t) \;  y^T(t) ]^T \in \RR^{n+q}$,  $C = [  I_n  \; 0_{n \times q}  ]$ and 
$ D =  [ 0_{ q \times n}  \;  I_{q}  ]$, then $x(t) = C z(t)$ and $y(t) = D z(t)$ for all $t \ge 0$.
SiDE (\ref{SiDE-xy})  can be written in a compact form  
\begin{subequations}    \label{Compact-z}
\begin{align}
  &  \d z(t) = F ( z(t), t) \d t + G (z(t), t) \d B(t), \quad   t \neq t_k     \label{SDE_z} \\
  &  \Delta_z (z_{t_k^-},     \xi (k-1), k  ) :=  z (t_k) - z (t_k^-)    \nonumber  \\
&     \;\; {} = H_F (z_{t_k^-},   k)   + \bH_G  (z_{t_k^-},    k ) \bar{ \xi } (k ) , \qquad \quad  k \in \NN \label{Impulse_z}
\end{align}
\end{subequations}
 with initial data $z(0) =z_0 =[ x_0^T  \; y_0^T ]^T$, where functions $F: \RR^{n+q} \times \RR_+ \to \RR^{n+q}$,  $G: \RR^{n+q} \times \RR_+ \to \RR^{(n+q) \times m}$, $H_F: C  ([t_{k-1}, t_k);\RR^{n+q})   \times \NN \to \RR^{n+q}$ and $\bH_G:C  ([t_{k-1}, t_k);\RR^{n+q})   \times \NN \to \RR^{(n+q) \times n}$ are given as
\begin{eqnarray*}
&&   F(z, t) = \begin{bmatrix}  f\left( C z, Dz, t \right)  \\
                      \tilde{f} \left( Cz  , D z, t \right) \end{bmatrix},  \;\;
 G(z, t) =\begin{bmatrix}  g\left( C z, Dz, t \right)     \\
                      \tilde{g} \left( Cz, D z , t  \right) \end{bmatrix},      \\
&&   H_F (z_{t_k^-},    k) = \begin{bmatrix}      0_{n \times 1}   \\
                      \tilde{h}_f \left( Cz_{t_k^-}, Dz_{t_k^-}, k \right)    \end{bmatrix},    \\
&&   \bH_G (z_{t_k^-},    k) = \begin{bmatrix}      0_{n \times n}   \\
                      \bar{h}_g \left( Cz_{t_k^-}, Dz_{t_k^-}, k \right)    \end{bmatrix}.
\end{eqnarray*}

Let us fix, for simplicity only, any $z(0) =z_0 = [ x_0^T  \;  y_0^T ]^T \in \RR^{n+q}  $. 
Obviously,  these functions obey $F(0, t)=0 $, $G(0, t)=0$, $H_F (0, k) =0$ and $\bH_G(0, k)=0$  for all $t \in \RR_+$ and $k \in \NN$. And they satisfy the local Lipschitz  condition and the linear growth condition, that is,  there is a constant $L_{z, \bn} >0$ for every integer $\bn \ge 1$ such that
\begin{multline} \label{localLipschitz-z}
     | F(z, t) - F(\tilde{z}, t) |^2 \vee | G(z, t) - G(\tilde{z}, t) |^2                \le L_{z, \bn}    |z- \tilde{z}|^2       \\
   | H_F(z_{t_k^-}, k) - H_F(\tilde{z}_{t_k^-}, k) |^2    \vee   | \bH_G (z_{t_k^-}, k) - \bH_G (\tilde{z}_{t_k^-}, k) |^2    \\
  \le L_{z, \bn}   ||z_{t_k^-}- \tilde{z}_{t_k^-} || ^2   
\end{multline}
 for all  $(z, \tilde{z}, z_{t_k^-}, \tilde{z}_{t_k^-}) \in \RR^{n+q} \times  \RR^{n+q} \times  C  ([t_{k-1}, t_k); \RR^{n+q}) \times C  ( [t_{k-1}, t_k); \RR^{n+q)} )$ with $|z| \vee | \tilde{z} | \vee || z_{t_k^-}|| \vee || \tilde{z}_{t_k^-} || \le \bn$,  $t \in \RR_+$ and $k \in \NN$; there is a  constant $L_z > 0$   such that
\begin{eqnarray} \label{lineargrowth-z} 
&&  | F(z, t)   |^2 \vee | G(z, t)  |^2     
 \le     L_z     |z |^2    \nonumber  \\   
&&  | H_F(z_{t_k^-}, k)   |^2   \vee  | \bH_G (z_{t_k^-}, k)   |^2  
  \le     L_z      ||z_{t_k^-}  || ^2   
\end{eqnarray}
 for all  $(z,   z_{t_k^-}) \in \RR^{n+q}     \times  C  ([t_{k-1}, t_k); \RR^{n+q})  $,  $t \in \RR_+$ and $k \in \NN$.
They are exactly the compact forms of  Assumption \ref{localLipschitz} and Assumption \ref{lineargrowth}, respectively. 
With Assumptions \ref{localLipschitz}-\ref{lineargrowth},  we have the existence and uniqueness of solutions to SiDE  (\ref{Compact-z}).
%
\begin{lemma}  \label{existence_n_uniqueness}
     Under Assumptions   \ref{localLipschitz}-\ref{lineargrowth}, there exists a unique (right-continuous) solution   to SiDE (\ref{Compact-z}),   denoted by 
$
z(t) = [ x(t)^T  \; y(t)^T ]^T = z(t; z_0) = [ x(t; x_0, y_0)^T  \; y(t; x_0, y_0)^T ]^T
$,
  and the solution belongs to  $\mm^2 ([0, T]; \RR^{n+q})$ for all $T \ge t \ge 0$, where $x(t) $ and $y(t)$ are continuous and    right-continuous processes, respectively.
\end{lemma}
The proof of Lemma \ref{existence_n_uniqueness} is relegated to Appendix. 
Now that we have   the existence and uniqueness of solutions to SiDE  (\ref{Compact-z}), or say, SiDE  (\ref{SiDE-xy}), we shall further study the stability of the unique solution of the SiDE. Let us introduce the definitions of exponential stability for  SiDE  (\ref{Compact-z}).

\begin{definition}   \label{pthM_ExpStab}
 \cite[Definition 4.1, p127]{mao2007book} The system (\ref{Compact-z}) is said to be $p$th ($p>0$) moment exponentially stable if there is a pair of positive constants $K$ and $c$ such that 
$ \EE |z(t)|^p \le K |z_0|^p e^{-c t }$ for all $t \ge 0 $,
 which implies
$ \limsup_{t \to \infty} \frac{1}{t}  \ln (\EE |z(t)|^p) \le -c  < 0   $
for all $z_0 \in \RR^{n+q}$.
\end{definition}
\begin{definition}   \label{AS_ExpStab}
 \cite[Definition 3.1, p119]{mao2007book} The system (\ref{Compact-z}) is said to be almost surely exponentially stable if 
$   \limsup_{t \to \infty} \frac{1}{t}  \ln |z(t)|  < 0  $
   for all $z_0 \in \RR^{n+q}$.
\end{definition}


\section{Lyapunov stability  of  synthetic CPSs}   \label{sec:stabilitySiDEs}

In this section, we  establish by the Lyapunov method a stability theory for the general class of SiDEs.  For simplicity, the compact form (\ref{Compact-z}) of CPS (\ref{SiDE-xy}) is used to study the existence and uniqueness of solutions to the SiDE. Here we  exploit the structure and study stability of the synthetic CPS  (\ref{SiDE-xy}).

\begin{theorem}   \label{Theorem_ExpStab}
Suppose that  Assumptions    \ref{localLipschitz}-\ref{lineargrowth}  hold and there is a pair of candidate Lyapunov functions $V \in C^{2, 1} (\RR^n \times \RR_+; \RR_+)$ and $\tilde{V} \in C^{2, 1}  (\RR^q \times \RR_+; \RR_+)$ for subsystems (\ref{SDE_x}) and  (\ref{SDE_y},\ref{Impulse_y}), respectively,  such that

{\rm (i)} for all $(x, y, t) \in \RR^n \times \RR^q \times \RR_+$ and some positive constants $ c_2 \ge  c_1>0$,  $\tilde{c}_2 \ge \tilde{c}_1>0$ and $p >0$,
\begin{subequations}      \label{LyapunovFunctions}
\begin{align}
 & c_1 |x|^p \le V(x, t) \le c_2 |x|^p,  \label{LyapunovFunctions_x} \\
&  \tilde{c}_1 |y|^p \le \tilde{V}(y, t) \le \tilde{c}_2 |y|^p ;  \label{LyapunovFunctions_y}
\end{align}
\end{subequations}
%

{\rm (ii)}  for all $(x, y ) \in \RR^n \times \RR^q $, some positive constants $  \al_1, \tilde{\al}_2$ and nonnegative $\al_2,   \tilde{\al}_1$,
\begin{subequations}        \label{LV-EV-xy}
\begin{align}
 &   \ll V(x, y, t) \le - \al_1 V(x, t) + \al_2  \tilde{V} (y, t),    \;\;  t \ge 0
  \label{LV_1}  \\ 
 &   \tll \tilde{V} (x, y, t)  \le \tilde{\al}_1 V(x, t)  + \tilde{\al}_2 \tilde{V} (y, t)    , \;\;  t \neq  t_k ;
   \label{tLtV_1}  
\end{align}
\end{subequations}
%

{\rm (iii)} at $t = t_k$ for  each $k \in \NN$,
\begin{eqnarray}
\lefteqn{   \EE   \tilde{V}   ( \phi (t_k^-)  + \tilde{\Delta} (\varphi_{ t_k^-}, \phi_{t_k^-}, k ), t_k )   }      \nonumber   \\
 &&         \le \tilde{\bt}_1  \sup_{t_{k-1} \le s < t_k} \EE     V ( \varphi (s), s    )    \nonumber   \\
 &&      {}  + \tilde{\bt}_2 \sup_{t_{k-1} \le s < t_k} \EE    \tilde{V}  ( \phi (s ), s    ) + \tilde{\bt}_3   \EE  \tilde{V}   ( \phi (t_k^-)  , t_k^- ) \;\; \label{EtV_tk-1}
\end{eqnarray}
      for all  $( \varphi_{ t_k^-}, \phi_{t_k^-} ) \in \lf_{\ff_t}^p ( [t_{k-1}, t_k); \RR^n) \times \lf_{\ff_t}^p ( [t_{k-1}, t_k); \RR^q) $, where   $\tilde{\bt}_1$,   $ \tilde{\bt}_2 $ and $ \tilde{\bt}_3$ are nonnegative constants such that 
\begin{equation}   \label{iii-constants}
 0 <  {\al_1}^{-1} \al_2  \tilde{\bt}_1 + \tilde{\bt}_2 + \tilde{\bt}_3 < 1 . 
\end{equation}
SiDE (\ref{Compact-z}), namley, CPS (\ref{SiDE-xy}) is $p$th moment exponentially stable provided that the impulse time sequence $\{ t_k \}_{k \in \NN}$ satisfies
\begin{equation}  \label{ImpulseInterval}
  0 < \underline{\Delta t} \le  \overline{\Delta t} <  \htau (\hatq) := \frac{- \ln \hatq } { ( {\al_1 \hatq } )^{-1}  \al_2    \tilde{\al}_1 + \tilde{\al}_2}
\end{equation}
for some $\hatq \in (  {\al_1}^{-1} \al_2    \tilde{\bt}_1 + \tilde{\bt}_2 + \tilde{\bt}_3, 1)$. 
  
\end{theorem}

{\begin{proof} According to Lemma \ref{existence_n_uniqueness}, that    Assumptions    \ref{localLipschitz}-\ref{lineargrowth}   hold  implies there exists a unique solution $z(t; z_0)$ to SiDE (\ref{Compact-z}) and the solution $z(t; z_0)$ belongs to  $\mm^2 ([0, T]; \RR^{n+q})$ for all $T \ge t \ge 0$.  By Lemma \ref{existence_n_uniqueness},  $x(t; x_0)$ is continuous on $[0, \infty)$ and $y(t; y_0)$ is  right-continuous on $[0, \infty)$ which could only jump at $ \{ t_k \}_{ k \in \NN } $.  Some ideas and techniques in this proof are derived from our results  \cite{huang202x,huang_partI}  as well as     \cite[Theorem 3.1 and Remark 3.1]{huang2009} on $p$th moment input-to-state stability (ISS)  of stochastic systems, see also   \cite{fridman2008,khalil2002}.

For notation, 
let $U(t) = \EE V(x(t), t ) $, $ W(t) = \EE \tilde{V}( y(t), t) $ for all $t \ge 0$ and, hence, $ || U_t || = \sup_{  t_\ast  \le \theta  \le t  } \EE V(x(\theta), \theta ) $, $ || W_t || = \sup_{  t_\ast  \le \theta  \le t  } \EE \tilde{V} (y(\theta), \theta ) $. So $U(t)$ is continous on $[0, \infty)$ and $W(t) $ is right-continuous on $  [0, \infty)$ and could only jump at     $\{ t_k \}_{k \in \NN}$;  $ || U_0  || = U(0)=  V(x_0, 0) $, $ || W_0  || = W(0)=  \tilde{V} (y_0, 0) $,  $ || U_{t_k}  || =U(t_k)= \EE V(x(t_k), t_k) $, $ || W_{t_k}  || =W(t_k)= \EE \tilde{V} ( y(t_k), t_k) $  for all $k \in \NN$ and $|| U_t || \ge  U(t) = \EE V(x(t), t )$, $|| W_t || \ge  W(t) = \EE  \tilde{V} ( y(t), t )$ for all $ t \ge 0$;   $ || U_t ||$ and $ || W_t ||$ are   continuous and nondecreasing on $[ t_{k-1}, t_k) $ and, hence, both they are right-continuous on $  [0, \infty)$ and could only jump at $\{ t_k \}_{k \in \NN}$.

%
The proof is so technical that we devide it into five steps, in which we will: 1) show the ISS of $x(t)$ with $y(t)$ as input; 2) combine the candidate Lyapunov functions $V(t)$ and $\tilde{V}( t) $ for the exponential stability of both $x(t)$ and $y(t)$;  3) construct a function that breaks the time interval into a disjoint union of subsets on which the system has  different  properties;  4) prove  the exponential stability of both $x(t)$ and $y(t)$; and 5) show the exponential stability of $z(t)$. 

{\it Step 1:}  
By  the It{\^o} formula and condition (\ref{LV_1}), 
\begin{multline*}
     U ( t )  = U (\bar{t} ) + \int_{\bar{t}}^t \EE \ll V(x(s), y(s), s) \d s  \\
   \le  U (\bar{t} ) + \int_{\bar{t}}^t \big[ - \al_1 U(s) + \al_2 W(s)  \big] \d s   \quad \forall \, t \ge \bar{t} \ge0
\end{multline*}
 and hence the upper right Dini derivative 
\begin{equation}    
  \dd^+ U(t)  = \EE \ll V(x(t), y(t), t)  \le - \al_1 U(t) + \al_2 W(t)       \label{dUdt-1}
\end{equation}
for all $t \ge 0$, which implies
\begin{equation}
 \dd^+ (t)   \le - (1 - \al) \al_1 U(t)      \;\;  {\rm if} \; \;   U(t) \ge \frac{ \al_2}{ \al_1 \al } \sup_{0 \le s \le t}   W (s)               \label{dUdt-2}
\end{equation}
  where $\al $ can be any positive   on $(0 ,   1)$. By \cite[Lemma 1]{fridman2008} and \cite[Theorem 4.18, p172]{khalil2002}, inequalities (\ref{LyapunovFunctions_x}) and (\ref{dUdt-2}) imply
\begin{equation}
  U(t) \le \Big( U(0 ) e^{- (1 - \al) \al_1  t } \Big)    \vee    \Big(  \frac{ \al_2}{ \al_1 \al}  \sup_{0 \le s \le t}   W (s)  \Big)  \label{U-bounded-t}
\end{equation}
 for all $t \ge 0$. 
If $\al_2 =0$,  $U(t)$ is  exponentially stable; otherwise (viz. $\al_2 >0$), $U(t)$ is ISS with  $W(t)$ as input, which means that  $x(t)$ is $p$th moment ISS with  $y(t)$ as input \cite{huang2009}. Specifically,  there is $t^U \ge 0$ (dependent on $U(0)$ and  $[( 1 - \al )  \al_1 ]^{-1} \al_2 \sup_{0 \le s \le t}   W (s)$,  see  \cite{fridman2008, khalil2002}) such that 
\begin{equation*}     
\begin{array}{ll}
    U(t) \le U(0 ) e^{-(1 - \al)  \al_1 t }, & \forall \, 0 \le t \le t^U     \\
   U(t)  \le   ( \al_1 \al )^{-1} \al_2   \sup_{0 \le s \le t^U} W (s) , &   \forall \, t \ge  t^U.  
\end{array}
\end{equation*}  
Moreover, $U(t)$ is (expoentially) stable if  $W(t)$ (exponentially) converges to zero as $ t \to \infty$, or say, if $y(t)$ is $p$th moent exponentially stable, so is $x(t)$  \cite[Theorem 3.1 and Remark 3.1]{huang2009}. Note that, if  $\al_2 =0$ and, hence,  (\ref{U-bounded-t}) implies that $U(t)$   is exponentially stable, Theorem \ref{Theorem_ExpStab} can be proved in a way similar to the proof of  \cite[Theorem 3.1]{huang_partI}. 
 It is easy to observe that \cite[Theorem 3.1]{huang_partI} is a specific case of Theorem \ref{Theorem_ExpStab}  with $\al_2 =0$. 
So  this proof  focuses on the case   $\al_2 >0$  in which $U(t)$ is ISS with  $W(t)$ as input.

{\it Step 2:}  
By conditions (\ref{iii-constants}) and (\ref{ImpulseInterval}), there exists a number $\hatq  \in ( {\al_1}^{-1} \al_2    \tilde{\bt}_1 + \tilde{\bt}_2 + \tilde{\bt}_3, 1 )$ for
 \begin{equation*}  
     [  ({\al_1 \hatq } )^{-1}  \al_2    \tilde{\al}_1+ \tilde{\al}_2 ] \, \overline{\Delta t}  < -\ln (\hatq) < -\ln ( {\al_1}^{-1} \al_2    \tilde{\bt}_1 + \tilde{\bt}_2 + \tilde{\bt}_3 ) .
\end{equation*}
This implies that one can find a pair of positive numbers
 $ \al \in (0,  1)$  sufficiently  close to $1$  for
\begin{equation}  \label{def-al}
       \big( \frac { \al_2    \tilde{\al}_1 } { \al_1  \al  \hatq }+ \tilde{\al}_2 \big)  \overline{\Delta t}   < -\ln (\hatq) < -\ln  \big(  \frac{ \al_2    \tilde{\bt}_1}{  \al   \al_1 }     + \tilde{\bt}_2 + \tilde{\bt}_3   \big )
\end{equation}
and  then    $\mu \in (0, ( 1 - \al) \al_1 \uDt / \oDt )$   
sufficiently small  for
\begin{multline}    \label{def-mu}
  \big( \frac { \al_2    \tilde{\al}_1 } { \al_1  \al \hatq }   +  \tilde{\al}_2    + \mu \big)  \oDt  < -\ln (\hatq)     \\
 < -\ln \Big(  \big( \frac{ \al_2    \tilde{\bt}_1}{  \al    \al_1 } + \tilde{\bt}_2 \big) e^{\mu \oDt } +  \tilde{\bt}_3  \Big). 
\end{multline}

Given   $\mu \in (0, (1 - \al) \al_1 \uDt / \oDt )$    by    (\ref{def-mu}),  
  let 
\begin{gather} 
   \tU (t) = e^{ \mu t} U(t)       \;\; {\rm and } \;\; 
   \tWa (t) =  e^{ \mu t}  W(t)   \label{def-tU-tW} 
\end{gather}
for all $ t \ge 0$. 
By    the It{\^o}   formula as well as (\ref{dUdt-1}) and  (\ref{tLtV_1}), 
\begin{eqnarray}
  \lefteqn{  \tU ( t )  =  \tU ( \bar{t}) +    \int_{\bar{t}}^t   e^{ \mu s} \big[ \mu U (s) +  \dd^+ U(s) \big]   \d s  }   \nonumber   \\
&&   \le   \tU ( \bar{t}) +    \int_{\bar{t}}^t  e^{ \mu s}  \big[ ( \mu - \al_1 ) U ( s) +  \al_2   W(s)    \big] \d s   \nonumber  \\
&&     =   \tU ( \bar{t}) +    \int_{\bar{t}}^t    \big[  - ( \al_1 - \mu )  \tU ( s) +   \al_2  \tWa (s)    \big] \d s \quad \label{EtU_t}
\end{eqnarray}
for all $t \ge \bar{t}  \ge 0$ and
\begin{eqnarray}
  \lefteqn{  \tWa ( t )= \tWa ( \tilde{t}) +     \int_{\tilde{t}}^t  e^{ \mu s} \big[ \mu W(s) +  \EE \tll \tilde{V}(x(s), y(s), s)  \big]  \d s   } \nonumber   \\
 &&  \le   \tWa ( \tilde{t}) +    \int_{\tilde{t}}^t  e^{ \mu s}  \big[ \tilde{\al}_1 U ( s) +    ( \tilde{\al}_2 + \mu )  W(s)    \big] \d s  \nonumber   \\
 && =  \tWa ( \tilde{t}) +    \int_{\tilde{t}}^t    \big[   \tal_1    \tU ( s) +  ( \tilde{\al}_2 + \mu)  \tWa (s)    \big] \d s  \hspace{1.36cm} \label{EtWa_t}
\end{eqnarray}
for all  $t_{k-1} \le \tilde{t} \le t < t_k$ and $k \in \NN$.  
For convenience, let
\begin{equation}   \label{def-tWa}
 \tW (t)=  \frac{ \al_2} { \al_1   \al} \tWa (t) = \frac{ \al_2} { \al_1   \al} e^{\mu t} W(t)
\end{equation} 
for all $t \ge 0$, where  $\al \in (0, 1) $ is given by (\ref{def-al}).

Let us define 
\begin{equation}  \label{def-oW}
\oW (t) = \tU(t) \, \vee  \tW (t) \quad \forall \, t \in [ 0, \infty).
\end{equation}
 Due to the continuity of $U(t)$ and the right-continuity of $W(t)$, $\oW ( t ) $ is right-continuous on $  [0, \infty)$ and could only   jump at the impulse instants $\{ t_k \}_{k \in \NN}$. 
Clearly, $   \oW(t) \ge \tU(t)$ and $   \oW(t) \ge \frac{ \al_2} { \al_1   \al}  \tWa (t) $ for all $ t \ge 0$. Recall that  $\al_2 >0$. 
So both $U(t)$ and $W(t)$ will be  exponentially stable if there is a positive constant $K$  such that
\begin{equation}    \label{oW-K}
    \oW (t)   < K  
\end{equation}
for all $t \ge t_0=0$. For instance,  
let 
\begin{equation} \label{def-K}
 K=    \frac{ \al_1 + \al_2} {   \al_1 \al    \hatq }      
 \big[  U(t_0) + W(t_0)  \big]  > 0
\end{equation}
 and hence  $ \oW (t_0) \le    U(t_0) +  \frac{  \al_2} {   \al_1  \al      }W(t_0)  < \hatq  K $.  

{\it Step 3:}  
Define function $\bar{v}: \RR_+ \to \RR$ by 
\begin{eqnarray}   \label{def-barv}
  \bar{v}(t) =      \tW (t)-    \tU (t)   \quad \forall \, t \in [0, \infty)
\end{eqnarray}
 with initial value 
$\bar{v}(0) = \frac{ \al_2   } {\al_1  \al  }  W(0)- U(0)$, where  $\al \in (0, 1)$ is given by (\ref{def-al}) and functions $\tU (t)$ and $\tW (t)$   by (\ref{def-tU-tW}) and (\ref{def-tWa}), respectively.
Since  $\tU(t)$ is continuous on $[0, \infty)$ and $\tW (t) $ is right-continuous on $  [0, \infty)$ and could only jump at  $\{ t_k \}_{k \in \NN}$, $ \bar{v}(t)$ is right-continuous on  $  [0, \infty)$ and could only  jump at the impulse instants $\{ t_k \}_{k \in \NN}$. Given any $t \ge 0$, either $\bar{v}(t) \ge 0$ or $ \bar{v}(t) < 0$. So the interval $[0, \infty)$ is broken into a disjoint union of subsets $T_+ \cup T_-$, where
\begin{equation}   \label{def-T+T-}
  T_+ = \{ t \ge 0:  \bar{v}(t) > 0 \} , \;\;   T_- = \{ t \ge 0:  \bar{v}(t) \le 0 \}.
\end{equation}

From (\ref{def-oW}), (\ref{def-barv}) and (\ref{def-T+T-}), 
\begin{equation}  \label{oW_T+T-}
 \oW (t) = \left\{ \begin{array}{cc}   \tW (t), & t \in T_+ \\  \tU(t), & t \in T_-  \end{array} \right.
\end{equation} 
and,  by (\ref{EtU_t})  and (\ref{def-T+T-}), 
\begin{equation}   \label{DtU_T-}
\dd^+ \tU (t) \le -  c \, \tU(t)  \quad  \forall \, t \in T_-   
\end{equation}
 where $c \in (0,   (1-\al ) \al_1 - \mu)$ is some postive number, e.g., $c =  [ (1-\al ) \al_1 - \mu  ]/2  $. That is, $\dd^+ \tU(t)$    is negative definite (with respect to $x$) and is strictly decreasing on the set $T_-$ if $T_- \neq \emptyset$.
 It is observed that $T_+ = \emptyset$ and, therefore,  $T_- = [0, \infty) $ if  $\al_2 =0$  
.  In fact,   $T_+ = \emptyset$, namely, $T_- = [0, \infty)$ implies that  $\dd^+ \tU (t) \le -  c\, \tU(t) $ for all $t \ge 0$ and hence   $U(t)$ is exponentially stable. 
In this case,  due to  $\tW (t) \le \tU(t)$  on $T_- = [0, \infty) $, both  $U(t)$ and  $\tW (t)$ are exponentially stable. 
Let us  consider the other case, namely, $T_+ \neq  \emptyset$.

Given any $t \in T_+$, due to the right-continuity of $\bar{v} (t)$ on $  [0, \infty)$, there exists an interval $[ \tau^+_1 (t), \tau^+_2 (t) )$ with $\tau^+_1 (t) < \tau^+_2 (t)$   such that $ (\tau^+_1 (t), \tau^+_2(t) ) \subset T_+$, where  
\begin{equation}  \label{def-tau_1-tau_2}
\begin{array} {c}
   \tau^+_1 (t) =  \inf \{ \, \bar{\tau}  \le t: \;  \bar{v}(\tau) >  0, \, \forall \tau \in [ \, \bar{\tau},t \, ] \},  \\
    \tau^+_2 ( t ) =      \sup \{  \, \bar{\tau} > t :  \; \bar{v}(\tau) > 0, \, \forall \tau \in [ \, t, \bar{\tau}) \}.
\end{array}
\end{equation}
%
Similarly, given any $\bart \in T_-$,  there is  an ordered pair   $ \tau^-_1 ( \bart ) \le \tau^-_2 ( \bart ) $ such that $ [ \tau^-_1 ( \bart ),  \tau^-_2 ( \bart ) ) \subset T_-$,  where
\begin{equation}  \label{def-tau-1-tau-2}
\begin{array} {c}
   \tau^-_1 ( \bart ) =  \inf \{ \, \tilde{\tau}  \le t: \;  \bar{v}(\tau) \le  0, \, \forall \tau \in [ \, \tilde{\tau},t \, ] \},  \\
    \tau^-_2 ( \bart ) =      \sup \{  \, \bar{\tau} \ge  t :  \; \bar{v}(\tau) \le  0, \, \forall \tau \in [ \, t, \bar{\tau}) \},
\end{array}
\end{equation}
and $[ \tau^-_1 ( \bart ),  \tau^-_2 ( \bart ) ) = \emptyset$   if $\tau^-_1 ( \bart ) =  \tau^-_2 ( \bart ) =\bart $. 

For convenience, we also write $\tau^+_1 =  \tau^+_1 (t ) $, $\tau^+_2 =  \tau^+_2 (t ) $, $  \tau^-_1 = \tau^-_1 (\bart ) $ and  $ \tau^-_2= \tau^-_2 (\bart ) $ where there is no ambiguity.

{\it Step 4:}  Let   us show   (\ref{oW-K})   for all $t \ge t_0 =0$. Define
\begin{equation}  
 \btauK = \inf \{ t \ge t_0 :   \oW (t) \ge K   \} , \label{def-htau} 
\end{equation}
By choice (\ref{def-K}),  $     \btauK > t_0 =0  $. 
If $\btauK > t_k$ for all $k \in \NN$, then (\ref{oW-K}) holds for all $t \ge 0$ because $\uDt = \inf_{k \in \NN} \{ t_k - t_{k-1} \} >0$ and $t_k \to \infty$ as $k \to \infty$.  
Otherwise, there is some   $k \in \NN$ such that $t_k = \inf \{ t_j: t_j \ge \btauK, j \in \NN \} $.  This means that either $ \btauK = t_k$ or $ t_{k-1} < \btauK < t_k$.  If $ \btauK = t_k$, then  (\ref{oW-K}) holds for all $t \in [0, t_k)$. Particularly,
\begin{equation}   \label{oW-upperbound-tk-}
   \tW (t_k^-) \le  || \tU_{t_k^-} ||  \vee || \tW_{t_k^-} ||  = || \oW_{t_k^-} || < K .
\end{equation}
Moreover,  either  $  \btauK = t_k \in T_+$ or $  \btauK = t_k \in T_-$ when $ \btauK = t_k$.
 If $  \btauK = t_k \in T_+$, then $\oW(t_k) = \tW ( t_k)  \ge K   $.  By condition (iii) with (\ref{def-mu}) and 
(\ref{oW-upperbound-tk-}), at each $t_k \le \btauK$, 
\begin{eqnarray}  \label{tW-htau=tk}
  \lefteqn{  \tW ( t_k)  =   \frac{ \al_2    } { \al_1  \al} e^{\mu t_k} W(   t_k )  }  \nonumber   \\
&&  \le     \frac{ \al_2    } { \al_1   \al} e^{\mu t_k}  \big( \tbt_1 || U_{t_k ^-} || + \tbt_2 || W_{t_k ^-} || + \tbt_3 W(t_k ^-)  \big)   \nonumber   \\
&&  \le     \big( \frac{ \al_2   \tbt_1  } { \al_1   \al}   || \tU_{t_k ^-} || + \tbt_2 || \tW_{t_k ^-} || \big) + \tbt_3 \tW(t_k ^-)  \nonumber   \\
&& \le  \Big[ \big( \frac{ \al_2   \tbt_1  } { \al_1   \al}    + \tbt_2 \big)  e^{\mu \oDt }  + \tbt_3 \Big]  || \oW_{t_k ^-} ||  \nonumber   \\
&&  <  \Big[ \big( \frac{ \al_2   \tbt_1  } { \al_1   \al}    + \tbt_2 \big)  e^{\mu \oDt }  + \tbt_3 \Big]   K    
< \hatq K  <   K ,
\end{eqnarray}
 which is a contradiction. So  $   t_k \notin T_+$ if $\btauK = t_k$.

If $  \btauK = t_k \in T_-$, then there are two possible cases:   $t_k^- \in T_-, \btauK = t_k \in T_-$ and    $t_k^- \in T_+, \btauK = t_k \in T_-$.

    Recall that $U(t) $ and hence $\tU (t)$ are   continuous on $[t_0, \infty)$. If $t_k^- \in T_-$, $t_k \in T_-$, then, by (\ref{def-tau-1-tau-2}), there is $\tau^-_1 = \tau^-_1 (t_k) < t_k$ such that $[\tau^-_1, t_k ] \subset T_-$.  By  (\ref{DtU_T-}), $ \tU (\tau^-_1) \ge \tU (t_k) e^{ c (t_k - \tau^-_1)} $. 
This  with $  \btauK = t_k $ produces
\begin{equation*}
 \tU (\tau^-_1) \ge \tU(t_k) e^{ c (t_k - \tau^-_1)} \ge   K  e^{ c (t_k - \tau^-_1)} >K.
\end{equation*}
But $  \btauK = t_k > \tau^-_1$ also means that $\tU (\tau^-_1) < K  $, which is a contradiction. Therefore, $t_k^- \notin T_-$ if $\btauK = t_k \in T_-$.

    If $t_k^- \in T_+$, $ \btauK = t_k \in T_-$, then, due to the fact that $\tU (t)$ is continuous $[t_0, \infty)$,
\begin{equation}
   \oW (t_k^-) =  \tW (t_k^-) > \tU(t_k^-) = \tU(t_k) \ge K .       \label{htau_tW_tk-}
\end{equation} 
 Recall that $\tW (t) $ and  $\oW (t) $ are continuous on $(t_{k-1}, t_k) $;   that $t_k^- \in T_+$ implies that,   by (\ref{def-tau_1-tau_2}),  there is $ \tau^+_1< t_k$ such that $(\tau^+_1, t_k ) \in T_+$. By  (\ref{htau_tW_tk-}),  there is   $ \tau \in (\tau^+_1, t_k ) $ so close to $t_k$ that $ \oW (\tau) =  \tW (\tau) >   U(t_k) \ge K $.  But this is in contradiction with $ \btauK = t_k > \tau $. Hence $t_k^- \notin T_+$ if $\btauK = t_k \in T_-$.

So $\btauK = t_k$ cannot be true. Let us proceed to check whether $ t_{k-1} < \btauK < t_k$ could be true or not. Recall that both $\tU(t)$ and $\tW(t)$ are continuous on $( t_{k-1}, t_k)$, which means that both $\oW (t)$ and $\bar{v} (t) $ are continuous on $( t_{k-1}, t_k)$. If $ t_{k-1} < \btauK < t_k$, then there are two cases:  { \bf c1) } $\bar{v} (\btauK) <0$, namely, $\oW (\btauK) = \tU (\btauK) \ge K$  and {\bf c2) } $\bar{v} (\btauK ) \ge 0$, namely, $\oW (\btauK) = \tW (\btauK) \ge K$ including the special case  $\oW (\btauK) = \tW (\btauK) = \tU(\btauK)  \ge K$ in which $\bar{v} (\btauK) =0$. 

\begin{itemize}

\item[ {\bf c1)} ]  Due to the continuity of  $\bar{v} (t) $ on  $( t_{k-1}, t_k)$ as well as (\ref{def-tau-1-tau-2}),   that    $\bar{v} (\btauK) <0$  implies that $\btauK \in T_-$ with $  \tau^-_1 (\btauK) < \btauK < \tau^-_2 (\btauK) $ and hence, by $ t_{k-1} < \btauK < t_k$,  there is $\tau = t_{k-1} \wedge \tau^-_1 (\btauK) < \btauK $ such that $[\tau, \btauK] \subset T_-$ and therefore (\ref{DtU_T-}) holds on  $[\tau, \btauK] $. But this yields
\begin{equation*}
  \tU (\tau) \ge \tU (\btauK) e^{ c (\btauK - \tau)} >  \tU (\btauK) \ge K,
\end{equation*} 
while $\btauK > \tau$ gives $ \tU (\tau) <K$.  The contradiction implies that $\bar{v} (\btauK) <0$ or say $\oW (\btauK) = \tU (\btauK) \ge K > \tW (\btauK)$ cannot be true with $t_{k-1} < \btauK < t_k$.

\item[ {\bf c2) } ]  Notice that $\tW( t_{k-1} ) < \hatq K $ due to  (\ref{tW-htau=tk}).  
Define
\begin{equation}    \label{def-tv}
  \tv (t) = \tW (t) - \hatq \tU(t)  \quad \forall \, t \in [0, \infty)
\end{equation}
with $\hatq \in (0, 1)$   given by (\ref{ImpulseInterval}). 
Similarly, $\tv(t)$ is continuous on $(t_{k-1}, t_k) $ for all $k \in \NN$ and the interval $[0, \infty)$ is broken into a disjoint union of subsets $\tT_+ \cup \tT_-$, where
$  \tT_+ = \{ t \ge 0:  \tv (t) > 0 \}$ and $   \tT_- = \{ t \ge 0:  \tv (t) \le 0 \} $.

From (\ref{def-barv}), (\ref{def-T+T-}) and (\ref{def-tv}),  
 it is   observed that   $T_+  \subset \tT_+$, $\tT_-  \subset T_-$ and, therefore,   (\ref{DtU_T-}) holds on $\tT_-  \subset T_-$. 
Notice that  $ t_{k-1} < \btauK < t_k$ and  $\bar{v} (\btauK ) \ge 0$ (namely, $\oW (\btauK) = \tW (\btauK) \ge K$) imply that $ \tv (\btauK) = \tW (\btauK) - \hatq \tU (\btauK) > \bar{v} (\btauK) = \tW (\btauK) - \tU (\btauK)  \ge 0$ and, hence, $\btauK \in T_+ \subset \tT_+$. 
As in (\ref{def-tau_1-tau_2}),  
 there is an ordered pair $ \ttau^+_1 = \ttau^+_1 (\btauK)   <  \ttau^+_2 =  \ttau^+_2( \btauK)$ such that $\btauK \in ( \ttau^+_1,  \ttau^+_2) \subset \tT_+$.  There are also two cases: i)  $  \ttau^+_1 \le t_{k-1}$ and  ii)  $  \ttau^+_1 > t_{k-1}$.

\begin{itemize}
 \item[ i) ]  That $  \ttau^+_1 \le t_{k-1}$ means  $[ t_{k-1}, t_k \wedge \ttau^+_2 ) \subset \tT_+$. Recall that, by  (\ref{tW-htau=tk}) , $ \tW( t_{k-1} ) < \hatq K $.

\item[ ii)] That $  \ttau^+_1 > t_{k-1}$ implies  $\tv (\ttau^+_1) =0$ due to the continuity of $\tv (t)$ on $( t_{k-1}, t_k) $. Therefore, $\tW(\ttau^+_1 ) = \hatq \tU (\ttau^+_1) < \hatq K$ since $\tU (t) < K $ for all $t < \btauK$. 
\end{itemize}

Let $\ttau = t_{k-1}   \vee    \ttau^+_1$,  then $\tW(\ttau ) < \hatq K$ and $\tU (t) \le \tW (t) / \hatq$ on $[ \ttau, t_k \wedge \ttau^+_2 ) \subset \tT_+$
It immediately follows from (\ref{EtWa_t}) and (\ref{def-mu}) as well as the Gronwall inequality that
\begin{eqnarray*}
 \lefteqn{  \tW (t)  \le  \tW (\ttau ) + \int_{\ttau }^t \big[ \frac{ \al_2 \tal_1 }{ \al_1 \al} \tU (s) + ( \tal_2 + \mu ) \tW (s) \big] \d s } \\
&&       \le  \tW (\ttau ) + \int_{ \ttau }^t  \big( \frac{ \al_2 \tal_1 }{ \al_1 \al  \hatq}   +   \tal_2 + \mu   \big)  \tW (s)   \d s   \\
&&    \le \tW (\ttau ) e^{ \big( \frac{ \al_2 \tal_1 }{ \al_1 \al  \hatq}   +   \tal_2 + \mu   \big)  ( t- \ttau ) }   \\
&&    <  \hatq K  e^{ \big( \frac{ \al_2 \tal_1 }{ \al_1 \al  \hatq}   +   \tal_2 + \mu   \big)  ( t_k - t_{k-1} ) } \\
&&   \le \hatq K  e^{ \big( \frac{ \al_2 \tal_1 }{ \al_1 \al  \hatq}   +   \tal_2 + \mu   \big)  \oDt }  < K    \hspace{3cm}
\end{eqnarray*}
for all $ t \in ( \ttau, t_k \wedge \ttau^+_2 )$, which is in contradiction with   $\bar{v} (\btauK) \ge 0$ for $t_{k-1} < \btauK < t_k $.

\end{itemize}
 
Therefore, neither $\btauK = t_k$ nor $t_{k-1} < \btauK < t_k $ could be true for any $k \in \NN$. 
So $\btauK > t_k $ for all $k \in \NN$ and, hence, (\ref{oW-K}) holds for all $t \ge 0$. By condition (i),
 this  implies  that
\begin{equation}      \label{ExpStab_x-y}
   \EE |x (t) |^p  \le \frac{c_2} {c_1} K e^{- \mu t} \;\; {\rm and} \; \; \EE |y (t) |^p  \le \frac{\al_1 \al \tc_2} {\al_2 \tc_1} K e^{- \mu t}
\end{equation}
for all $t \ge 0$, where $\mu >0$ and $K>0$ are given by (\ref{def-mu}) and (\ref{def-K}), respectively.

{\it Step 5:} We have shown by (\ref{ExpStab_x-y})  the $p$th moment exponential stability of  $x(t)$ and that of $y(t)$.
Note that 
 $ z (t) = [ x^T(t) \;  y^T(t) ]^T $ 
and, hence,
\begin{equation*}    \label{z_Euclidean}
    | x(t) |^2 \vee | y (t) |^2 \le | z(t) |^2  = |x(t) |^2 + |y(t)|^2     
\end{equation*}
  for all $t \ge 0$. By the elementary and  the H{\" o}lder  inequalities,  
\begin{eqnarray}   \label{z_p}
  (  | z(t) |^2)^{p/2} 
   = ( |x(t) |^2 + |y(t)|^2 )^{p/2}   \nonumber   \\
  \le   k_p  ( |x(t) |^p + |y(t)|^p )
\end{eqnarray}
for all $ t \ge 0$,  where $k_p=1$ when $0<p \le 2$ and $k_p = 2^{(p-2)/2}$ when $ p > 2$. 
From (\ref{ExpStab_x-y}) and (\ref{z_p}), it follows that
\begin{eqnarray*}    
   \EE |z (t)|^p  \le k_p \EE |x(t)|^p + k_p  \EE |y(t)|^p     \hspace{1.8cm}    \\
  \le  \big( \frac{  c_2} {c_1} + \frac{\al_1 \al \tc_2} {\al_2 \tc_1}  \big) k_p K e^{- \mu    t}     \hspace{1.8cm}    \\
\le  \big( \frac{  c_2} {c_1} +\frac{\al_1 \al \tc_2} {\al_2 \tc_1}  \big) K_0 |z_0|^p e^{- \mu    t}  \quad \forall \, t \ge 0
\end{eqnarray*}
 where $K$ is given by (\ref{def-K}) and 
$K_0 = \frac{ \al_1 + \al_2} {   \al_1 \al    \hatq }  (c_2 + \tc_2) k_p  $.

This means that SiDE  (\ref{Compact-z}) (viz. CPS (\ref{SiDE-xy})), or say, $z(t)$ is $p$th moment exponentially stable (with Lyapunov exponent no larger than $- \mu$ and $\mu >0$ given by (\ref{def-mu})).  
\end{proof}
%
%
\begin{remark}  \label{remark-ImpulseInterval}
If $\al_1, \al_2, \tal_1, \tal_2 $ are  all positive and   determined,  condition (\ref{ImpulseInterval}) in Theorem \ref{Theorem_ExpStab} can be specified as
\begin{equation}  \label{ImpulseInterval_max}
0 < \uDt \le \oDt <    \htau (\hatq_\ast   \vee  \hatq_0) ,
\end{equation}
where   $\hatq_\ast$ and $\hatq_0$ are given by  (\ref{def-hatq-ast}) and (\ref{def-hatq-0}) below, respectively. 
Obviously,    $\htau (\hatq) > 0$ for every $\hatq \in (0,1)$ and $\htau (\hatq)$  is   a continuously differentiable function    on $(0,1)$ with derivative
\begin{equation}   \label{dhtau_dhatq}
 \frac{ \d \htau ( \hatq ) } { \d \hatq} = - \big( \frac{ \al_2 \tal_1   } {  \al_1  \sqrt{\tal_2}}  + \sqrt{\tal_2} \hatq \big)^{-2}  \tau^{\prime} (\hatq), 
\end{equation}
where 
$ 
\tau^{\prime} (\hatq) =  \frac{ \al_2 \tal_1   } {  \al_1   \tal_2 } \big( 1 + \ln   \hatq   \big) + \hatq
$.  
Note that   $\tau^{\prime} (\hatq) $ is  increasing  on $(0,\infty)$ and
the maximum   of $\htau (\hatq)$   is achieved at $\hatq=\hatq_\ast$ by
\begin{equation}    \label{def-hatq-ast}
  \tau^{\prime} (\hatq_\ast) = \frac{ \al_2 \tal_1   } {  \al_1   \tal_2 } \big( 1 + \ln   \hatq_\ast   \big) + \hatq_\ast =0
\end{equation}
and  $\hatq_\ast \in (e^{-(\al_1   \tal_2 + \al_2 \tal_1 ) /  (\al_2 \tal_1 ) }, 1)$ since 
$ \tau^{\prime} (1) = \frac{ \al_2 \tal_1   } {  \al_1   \tal_2 }  +1 >0 >  \tau^{\prime} (e^{-(\al_1   \tal_2 + \al_2 \tal_1 ) /  (\al_2 \tal_1 ) }) $.
One can compute  $\hatq_\ast$   by solving   (\ref{def-hatq-ast}) with the initial guess 
\begin{equation}   \label{def-hatq-0}
\hatq_0 =  ( {\al_1}^{-1} \al_2    \tilde{\bt}_1 + \tilde{\bt}_2 + \tilde{\bt}_3 ) \vee e^{-(\al_1   \tal_2 + \al_2 \tal_1 ) /  ( \al_2 \tal_1 ) }.
\end{equation} 
It is observed  from condition (\ref{ImpulseInterval}) of Theorem \ref{Theorem_ExpStab} that, for expoonential stability of system  (\ref{SDE_x}-\ref{Impulse_y}),     the choice of $\hatq$ is confined to  $\hatq \in ( \hatq_0, 1)$. By (\ref{dhtau_dhatq}) and (\ref{def-hatq-ast}) as well as (\ref{def-hatq-0}), 
\begin{equation*}
\sup_{\hatq \in (\hatq_0, 1)} \htau( \hatq) = \left\{
\begin{array}{cc}
 \htau (\hatq_\ast), & 0< \hatq_0 \le \hatq_\ast <1  \\
\htau (\hatq_0), & 0< \hatq_\ast < \hatq_0 <1
\end{array}   \right.
\end{equation*}
which implies that (\ref{ImpulseInterval_max}) exactly means 
\begin{equation}   \label{ImpulseInterval_sup}
    0 < \uDt \le \oDt < \htau (\hatq_\ast   \vee  \hatq_0) =\sup_{\hatq \in (\hatq_0, 1)} \htau( \hatq).
\end{equation}
Recall that $\htau (\hatq) $ is continuously differentiable  on $(0,1)$.  If (\ref{ImpulseInterval_sup}) holds, there is   $\hatq \in (\hatq_0, 1)$ sufficiently close to $\hatq_\ast   \vee  \hatq_0$   for (\ref{ImpulseInterval}).
\end{remark}

Furthermore,  under  the linear growth condition (Assumption \ref{lineargrowth}),  the  $p$th  moment exponential stability of SiDE (\ref{Compact-z}) implies its  almost sure exponential stability. The proof is similar to that of \cite[Theorem 4.2, p128]{mao2007book} and is omitted.

\begin{theorem}   \label{Theorem_AlmostSure}
Under Assumption \ref{lineargrowth},
  the  $p$th ($p>0$) moment exponential stability of  SiDE (\ref{Compact-z})  implies that it is also   almost surely exponentially stable.
\end{theorem}

\section{Stability  of sampled-data control systems}

Let us consider a sampled-data control system  
\begin{equation} \label{sampledSDE}
\d x(t) = [ \barf(x(t) ) + \bu(x(t_\ast) ) ] \d t + \barg( x(t) ) \d B(t) \quad t  \ge 0
\end{equation}
with initial value $x(0) =x_0 \in \RR^n$ and sampling sequence $\{ t_k \}_{k \in \NN}$,   where  $\barf  : \RR^n  \to \RR^n$ and $\barg : \RR^n   \to \RR^{n \times m}$   are measurable  functions with $\barf (0) =0$ and $\barg (0) =0$,  which both  satisfy the local Lipschitz condition and the linear growth condition, that is, 
 there is $\bar{L}_{\bn} >0$ for every integer $\bn \ge 1$ such that
$   | \barf (x ) - \barf (\bar{x} ) |^2   \vee  | \barg (x ) - \barg (\bar{x} ) |^2  \le \bar{L}_{\bn} | x - \bar{x} |^2   $
for all $(x, \bar{x}) \in \RR^n \times \RR^n$ with $|x| \vee |\bar{x}| \le \bn$ and there is $\bar{L} >0$ such that
$   | \barf (x )  |^2   \vee  | \barg (x )  |^2  \le \bar{L} | x  |^2  $
for all $ x \in \RR^n$; $\bu \in C^2 (\RR^n ; \RR^n)$ with  $\bu (0)=0$ is the control input.  Let  $y(t) = u(x(t) ) - u(x(t_\ast) )$ for all $t \ge 0$, then  $\d y(t) = \d u (x(t)) $ on $  ( t_{k-1}, t_k)$ and $y(t_k) - y(t_k^-) = u (x (t_{k-1}) )- u (x (t_k) )$ for all $k \in \NN$. 
By   the It{\^o}  formula,  one can derive  a cyber-physical model of the form  (\ref{SiDE-xy})  for sampled-data control system (\ref{sampledSDE}).   

In this paper, 
 we consider sampled-data  system (\ref{sampledSDE})  that has a linear feedback control
   $\bu (x) = \bB x$ with  matrix $\bB \in \RR^{n \times n}$ 
\begin{equation}    \label{sampledSDE-bB}
\d x(t) = [ \barf(x(t) ) + \bB x (t_\ast ) ] \d t + \barg( x(t) ) \d B(t)  \quad  t  \ge 0
\end{equation}
 so that not only can it  be  easily implemented  \cite{qian2002,tsinias1991} but also its cyber-physical model in the form of  CPS (\ref{SiDE-xy}) satisfies Assumptions \ref{localLipschitz}-\ref{lineargrowth}. 
%
%
Let $y(t) = x(t)  - x(t_\ast) $ for all $ t \ge 0$.
This implies that $\d y(t) = \d  x(t) $ on $  (t_{k-1}, t_k)$ and   
$y(t_k) =0$  for all $k \in \NN$. Using the It{\^o}  formula,  we obtain  a cyber-physical model  of sampled-data control system (\ref{sampledSDE-bB})  
\begin{subequations}  \label{SiDE-bB}
\begin{align}
& \d x(t) = [ \barf (x(t) ) + \bB ( x(t)   - y(t) ) ] \d t + \barg( x(t) ) \d B(t),  \nonumber   \\
& \hspace{5cm} \;\; {}  t \in [0, \infty)   \label{SDE-bB_x} \\
& \d y(t)  =   [ \barf (x(t) ) + \bB  ( x(t)   - y(t) ) ] \d t  +    \barg (x(t)) \d B(t),  \nonumber   \\ 
& \hspace{4.5cm} {} t \in [ 0, \infty) \setminus \{t_k\}_{k \in \NN}  \label{SDE-bB_y}  \\
&  
y(t_k) - y(t_k^-)     =  x (t_{k-1}) -  x (t_k^-),    \quad  k \in \NN   \label{Impulse-bB_y} 
\end{align}
\end{subequations}
with   $x(0)=x_0 \in \RR^n$ and $y(0) =0$.  
Clearly, the CPS  (\ref{SiDE-bB}) of  sampled-data control system (\ref{sampledSDE-bB}) is a specific case of  CPS  (\ref{SiDE-xy}) which satisfies Assumptions \ref{localLipschitz}-\ref{lineargrowth},  where 
$ f(x, y, t) =\tilde{f} (x, y , t)  = \barf (x) + \bB (x - y) $,  
  $  g(x, y, t) =  \tilde{g} (x, y , t) =\barg (x) $, 
 $ \tilde{h}_f (x_{t_k^-}, y_{t_k^-}, k) =x (t_{k-1}) -  x (t_k^-)   $,  
 and $ \bar{h}_g (x_{t_k^-}, y_{t_k^-},  k)  = 0$ for all $ t \in \RR_+$ and $ k \in \NN$. 
 Theorem \ref{Theorem_ExpStab}  and  Theorem \ref{Theorem_AlmostSure}    immediately  yield the following result (see also Remark \ref{remark-ImpulseInterval}).

\begin{theorem}    \label{Theorem_bB}
Suppose that    conditions 
(\ref{LyapunovFunctions})-(\ref{iii-constants})  
hold for  CPS  (\ref{SiDE-bB}). If the sampling sequence $\{ t_k \}_{k \in \NN}$ satisfies (\ref{ImpulseInterval}), then CPS (\ref{SiDE-bB}) is $p$th moment exponentially stable and is also  almost surely exponentially stable.
\end{theorem}

\begin{remark}   \label{remark-hybrid} 
The dynamics of a sampled-data system is   written as an impulsive system in the references   \cite{nagh2008,rios2016} too.  Note that some   approaches \cite{fridman2004,fridman2010,seuret2012} describe the sampled state $x(t_\ast)$ with input delay mechanisms  while the  hybrid system \cite[Eq.(13)]{nagh2008}  just depicts its subsystem  $x(t_\ast)$  as a pure jump process.  Clearly,  our cyber subsystem (\ref{SDE-bB_y},\ref{Impulse-bB_y}) is distinct from the pure jump process of  $x(t_\ast)$  in the literature. 
\end{remark}

\subsection{Controller emulation (Process-oriented models)}

By approach of controller emulation  that is from the viewpoint of process-oriented models,    a continuous-time state-feedback controller  is designed based on the continuous-time plant model 
for stability of the closed-loop system  
\begin{equation}      \label{controlledSDE}
\d x(t) = \barf_u (x )  \d t + \barg ( x(t) ) \d B(t)    \quad t  \ge 0
\end{equation}
with $\barf_u (x ) = \barf (x) + \bu (x) = \barf (x) + \bB x$ (being the drift of the closed-loop system) and then  the state-feedback  controller is discretized and implemented using a sampler and  ZOH device. This leads to the sampled-data control system (\ref{sampledSDE-bB}) and its cyber-physical dynamics is described by (\ref{SiDE-bB}). 
The main question in the design method   is, see  \cite{astrom1997,fridman2010,nagh2008,nesic2001},

{\it  for what sampling sequence $ \{ t_k \}_{k \in \NN} $ does the sampled-data control system  (\ref{sampledSDE-bB}) preserve the stability property of the  continuous-time  system   (\ref{controlledSDE})?  }

Let us apply our CPS theory   and address the main  question.
Specifically,  by Theorem \ref{Theorem_bB},  we find the conditions on   $ \{ t_k \}_{k \in \NN} $  for   exponential stability of the sampled-data system  (\ref{sampledSDE-bB}) when  the feedback control $\bu (x) = \bB x$ is designed such that
 \begin{equation}      \label{bLV-bB} 
     \ll V(x )    \le  - 2 \bal V(x) \quad    \forall  \, x \in \RR^n
\end{equation}
  and the closed-loop system (\ref{controlledSDE}) is exponentially stable  \cite{khas2012,mao2007book},
where $\bal >0 $ is a  constant, $V \in C^2 (\RR^n ; \RR_+)$ is a Lyapunov function  with (\ref{LyapunovFunctions_x}) and its infinitesimal generator $\ll V: \RR^n  \to \RR$  associated with  system (\ref{controlledSDE}) is, as (\ref{LV}) above,  
\begin{equation}    \label{bLV}
 \ll V (x) = V_x (x) \barf_u (x )  + \frac{1}{2}   {\rm trace} \left[ \barg^T (x ) V_{xx} (x) \barg (x) \right] .
\end{equation}

Let us first employ the same Lyapuov function  $V(x) = \tV (x) $ for both the physical and the cyber subsystems  since it is very helpful for  exposing not only  the  interactions between the subsystems \cite{huang_partI}  but also the intrinsic relationship between the two main approaches, see  Remarks \ref{ImpulseInterval-ce-rewrite}-\ref{equivalence-ce-dta}.

\begin{theorem}   \label{Theorem-ce-1}
Suppose that the Lyapunov function $V(x)$  with condition (\ref {bLV-bB})  for physical  system  (\ref{controlledSDE}) is a quadratic function  
\begin{equation}   \label{LyapunovFunction-quadratic}
    V (x) = x^T P x 
\end{equation}
 with matrix $P >0 $.
Let the sampling sequence $\{ t_k \}_{k \in \NN}$ satisfy 
\begin{equation}  \label{ImpulseInterval-ce}
  0 < \underline{\Delta t} \le  \overline{\Delta t} <  \htau (q_\ast)  
\end{equation}
where function $\htau : (0, 1) \to \RR_+$ is defined by
\begin{multline}    \label{def-htau-ce}
\htau (q)  = -   \bal \sqrt{q_\ast}   (\bal \sqrt{q} )^2  \ln q \cdot    \Big\{ \sqrt{ \bal_b} \big[ \big( 2  \bal \sqrt{q_\ast}  + \bal +   \sqrt{\bal_f } \big) \\
\hspace{-0.2cm} {} \cdot   ( \bal \sqrt{q} )^2   
  +  ( \bal  + \sqrt{\bal_f } ) ( \bal \sqrt{q_\ast} )^2    + 2  \sqrt{  \bal_b \bal_f }  \bal \sqrt{q_\ast}  \big]  \Big\}^{-1}
\end{multline}
and $q_\ast  \in (0, e^{-1} )$ is the unique root of 
\begin{multline}    \label{def-btau-prime-ce}
  \btau^{\prime} (q) := 2 ( \bal \sqrt{q} )^2 + (\bal + \sqrt{ \bal_f} ) \bal \sqrt{q}  \\
 {} +  \big[ ( \bal + \sqrt{ \bal_f } ) \bal \sqrt{q} + 2 \sqrt{ \bal_b \bal_f } \, \big] ( \ln q +1 )   =0 
\end{multline}
with $\bal_b >0$ and $\bal_f>0$ being   such that, for all $x \in \RR^n$,
\begin{equation}        \label{control-drift}
 V(\bB x) \le \bal_b V(x)   \;\;  {  and} \;\;   V \big( \barf_u (x) + \bal x \big) \le \bal_f V(x). 
\end{equation}
  Then   CPS  (\ref{SiDE-bB})   is mean-square exponentially stable and  is also almost surely exponentially stable, which implies that   its subsystem (\ref{SDE-bB_x}), viz.,   (\ref{sampledSDE-bB})      is mean-square exponentially stable and  is also almost surely exponentially stable.
\end{theorem}
\begin{proof}  It will follow from Theorem \ref{Theorem_bB} that  CPS  (\ref{SiDE-bB})   is mean-square exponentially stable and   also almost surely exponentially stable  if conditions  (\ref{LyapunovFunctions})-(\ref{ImpulseInterval})    with $p=2$ hold for   (\ref{SiDE-bB}). 

Let  $\tV (x)=V(x) $   defined as (\ref{LyapunovFunction-quadratic}). So
$\lambda_m (P) |x|^2 \le V(x) \le \lambda_M (P) |x|^2$ for all $x \in \RR^n$ 
 and hence condition (\ref{LyapunovFunctions}) holds with positives 
$p=2$, $c_1 =\tc_1=\lambda_m (P)$ and $c_2 =\tc_2 = \lambda_M (P) $.  

Since both $\barf (x)$ and $\bu (x) + \bal x = (\bB + \bal I_n )x$ satisfy the linear growth conditions $| \barf (x) |^2 \le \bar{L} |x|^2$ and $| ( \bB + \bal I_n )x|^2 \le | \bB + \bal I_n |^2 |x|^2$, so does $\barf_u (x) + \bal x $, that is,
$| \barf_u (x) + \bal x |^2 = | \barf (x) + (\bB + \bal I_n) x |^2 \le 2 ( \bar{L} + | \bB + \bal I_n |^2) |x|^2$ for all $x \in \RR^n$. Therefore, for all $x \in \RR^n$,   
\begin{eqnarray*}
&& V ( \barf_u (x) + \bal x ) \le \lmd_M (P) | \barf_u (x) + \bal x|^2  \\
&&  \hspace{2.3cm}   \le 2 (  \bar{L} + |\bB + \bal I_n |^2)  \frac{ \lmd_M (P) }{ \lmd_m (P) } V (x),  \\
&& V ( \bB x )   \le    \lmd_M (P) | \bB|^2 | x|^2 \le   \frac{ | \bB|^2 \lmd_M (P) }{ \lmd_m (P) } V (x).
\end{eqnarray*}
So there exist   positive constants   $\bal_b \in (0,  | \bB |^2 \lmd_M (P) / \lmd_m (P) ]$ and $\bal_f \in (0, 2 ( \bar{L} + | \bB + \bal I_n |^2)   \lmd_M (P) / \lmd_m (P) ]$  for (\ref{control-drift}). 

By  (\ref {bLV-bB}), (\ref{bLV})    and \cite[Lemma 1]{huang2009a}, for all $x, y \in \RR^n$,
\begin{eqnarray*}    
 \lefteqn{   \ll V (x, y ) =   2 x^T P [ \barf (x) + \bB x - \bB y ]  + {\rm trace} \big[ \barg^T (x) P \barg (x)   \big]  }   \\
 &&  =   2 x^T P \barf_u  (x)    + {\rm trace} \big[ \barg^T (x) P \barg (x)   \big]  - 2x^T P \bB  y  \hspace{1cm} \\
 &&  \le - 2 \bal V(x) - 2x^T P \bB  y   \\
 && \le -2 \bal   V(x)  + \bal  V(x) + \frac{1} {\bal}  V (\bB y)   \\
 && \le -   \bal    V(x) +  \frac{   \bal_b} {\bal} V (  y) .
\end{eqnarray*}
 Hence (\ref{LV_1}) holds with $\al_1 =   \bal   $ and $\al_2 =     \bal_b / \bal $.  Similarly, 
\begin{eqnarray*}
\lefteqn{   \tll V (x, y ) =   2 y^T P [ \barf (x) + \bB x - \bB y ]  + {\rm trace} \big[ \barg^T (x) P \barg (x)   \big]  }   \\
&&  =   2y^T P \barf_u (x)   + {\rm trace} \big[ \barg^T (x) P \barg (x)   \big]   - 2 y^T P \bB y \\
&& =   2y^T P [ \barf_u (x) + \bal x] + {\rm trace} \big[ \barg^T (x) P \barg (x)   \big]  \\
&& \quad  {} -2 \bal y^T P x - 2 y^T P \bB y    \\
&&   \le 2x^T P \barf_u (x)   + {\rm trace} \big[ \barg^T (x) P \barg (x)   \big]  + 2 \bal V(x)   \\ 
&& \quad  {} + 2(y-x)^T P [ \barf_u (x) + \bal x] -2 \bal y^T P x- 2 y^T P \bB y   \quad   \\
&&   \le     b_1^{-1} V (\barf_u (x) + \bal x )  + b_1 (y-x)^T P (y-x) \\
&& \quad    {}  -2 \bal y^T P x  + \sqrt{\bal_b} V(y) + \frac{1}{\sqrt{\bal_b}} V (\bB y)   \\
&& \le   b_1  V (x) + b_1^{-1}  V (\barf_u (x) + \bal x ) - 2 ( b_1 + \bal) y^T P x \\
&&  \quad   {} + (   \sqrt{\bal_b}+ b_1   ) V(y) + \frac{1}{\sqrt{\bal_b}} V (\bB y)
\end{eqnarray*}
\begin{multline}      \label{LV_2-ce}
\le ( b_1 +  \bal_f b_1^{-1}   ) V (x)    - 2 ( b_1 + \bal) y^T P x
      + (  2  \sqrt{\bal_b} + b_1  ) V(y)   \\
\le [ b_1 +  \bal_f b_1^{-1}  + ( b_1 + \bal) b_2^{-1} ] V (x)    \hspace{2.8cm}    \\
   {} + [  2  \sqrt{\bal_b} + b_1  + ( b_1 + \bal) b_2] V(y)     \hspace{1.7cm}
\end{multline}
for all $x, y \in \RR^n$, where $b_1  $ and $b_2 $    are positive constants to be determined. So condition (\ref{tLtV_1})  holds with $\tal_1 = b_1 +  \bal_f b_1^{-1}   + ( b_1 + \bal) b_2^{-1} $ and $\tal_2 = 2  \sqrt{\bal_b} + b_1  + ( b_1 + \bal) b_2 $.

Observe that  (\ref{Impulse-bB_y}) and $y(t) = x(t) - x (t_\ast)$ for all $t \ge 0$ give
$ y(t_k) =y(t_k^-) + x(t_{k-1} ) - x(t_k^-)  =0 $ for all  $  k \in \NN$.
This immediately produces $V( y(t_k)) =0$ and  (\ref{EtV_tk-1})  with $\tilde{\bt}_1 =\tilde{\bt}_2 = \tilde{\bt}_3 =0$, which implies that  nonngegatives  $\tilde{\bt}_1$,   $ \tilde{\bt}_2 $ and $ \tilde{\bt}_3$ can be chosen for  (\ref{iii-constants}) with  arbitrary small $ {\al_1}^{-1} \al_2    \tilde{\bt}_1 + \tilde{\bt}_2 + \tilde{\bt}_3  >0$. Therefore, conditions (\ref{EtV_tk-1})-(\ref{iii-constants})  hold. 

Since  $ {\al_1}^{-1} \al_2    \tilde{\bt}_1 + \tilde{\bt}_2 + \tilde{\bt}_3  >0$ can be arbitrary small, substitution of $\al_1 =   \bal   $, $\al_2 =   \bal_b / \bal $, $\tal_1 = b_1 + \bal_f b_1^{-1}    + ( b_1 + \bal) b_2^{-1} $ and $\tal_2 = 2  \sqrt{\bal_b} + b_1  + ( b_1 + \bal) b_2 $ into  (\ref{ImpulseInterval}) yields function $ \htau (q) = \btau (q, b_1, b_2) $ for $q \in (0, 1)$ with positive parameters $b_1, b_2$ to be determined, where function $\btau: (0,1) \times \RR_+ \times \RR_+ \to \RR_+ $ is defined by 
\begin{multline}  \label{def-btau-ce}
   \btau (q, b_1, b_2)  = - \bal^2 q  \ln q  \Big\{      \big[ 2  \sqrt{\bal_b} + b_1 +  ( b_1 + \bal ) b_2 \big]  \bal^2 q  \\
     {}     +   \bal_b  \big[ b_1 + \bal_f b_1^{-1}   + ( b_1 + \bal ) b_2^{-1}   \big]  \Big\}^{-1}.
\end{multline} 

The supremum $\sup_{ q \in (0,1) } \htau (q)$ in condition  (\ref{ImpulseInterval}) (see also Remark \ref{remark-ImpulseInterval}) can be obtained by solving  optimization problem
\begin{eqnarray}
&& \min \;     \btau^{-1} (q, b_1, b_2)      \label{min_btau-1}  \\
&&  {\rm s.t.} \;\; h_j (q, b_1, b_2) > 0, \quad j=1, 2, 3, 4  \nonumber 
\end{eqnarray}
where function $\btau^{-1} : (0,1) \times \RR_+ \times \RR_+ \to \RR_+ $ is given by 
\begin{equation}    \label{def-btau-1-ce}
 \btau^{-1} (q, b_1, b_2)  = \frac{1}{ \btau (q, b_1, b_2) }    
\end{equation}
with $ \btau (q, b_1, b_2)$ by (\ref{def-btau-ce})  and $h_j (q, b_1, b_2)$ is the $j$th element of vector $h (q, b_1, b_2) = \begin{bmatrix} q & 1-q & b_1 & b_2  \end{bmatrix}^T$ for $j=1, 2, 3, 4$. The Lagrangian $\lagr : \RR^3 \times \RR^4 \to \RR$ associated with the problem (\ref{min_btau-1}) is defined as, see, e.g., \cite{boyd2004},
\begin{equation}    \label{Lagrangian}
  \lagr (q, b_1, b_2, \lmd) = \btau^{-1} (q, b_1, b_2) - \lmd^T h (q, b_1, b_2)
\end{equation}
where  $ \lmd = \begin{bmatrix} \lmd_1 & \lmd_2 & \lmd_3 & \lmd_4 \end{bmatrix}^T $ is the Lagrangian multiplier vector. 
The Karush-Kuhn-Tucker (KKT) conditions give 
\begin{multline*}
 \frac{ \partial  \lagr (q, b_1, b_2, \lmd)} { \partial q} =   
 \frac{ \partial  \lagr (q, b_1, b_2, \lmd)} { \partial b_1} =    
  \frac{ \partial  \lagr (q, b_1, b_2, \lmd)} { \partial b_2} =0,  \\
h_j (q, b_1, b_2) > 0, \;  \lmd_j \ge 0, \;  \lmd_j h_j (q, b_1, b_2) =0, \;  j=1, 2, 3, 4
\end{multline*} 
which   imply $\lmd_j = 0$ for $ j=1, 2, 3, 4$. So  the Lagrangian (\ref{Lagrangian}) leads to  $ \lagr (q, b_1, b_2, \lmd) =  \btau^{-1} (q, b_1, b_2) $ and the KKT    optimality conditions  for the problem (\ref{min_btau-1})  
\begin{equation*}
 \frac{ \partial  \btau^{-1} (q, b_1, b_2) } { \partial q} =   
 \frac{ \partial  \btau^{-1} (q, b_1, b_2) } { \partial b_1} =  
  \frac{ \partial \btau^{-1} (q, b_1, b_2) } { \partial b_2} =0.   
\end{equation*} 
By  (\ref{def-btau-1-ce}) and  (\ref{def-btau-ce}), the KKT optimality conditions 
produce
\begin{eqnarray}
  \lefteqn{  \frac{ \partial \btau^{-1} (q, b_1, b_2) } { \partial b_2} = 0  \; \;  \Rightarrow \; \;  
 \bal^2 ( b_1 + \bal ) q -  \frac{ \bal_b ( b_1 + \bal ) } { b_2^ 2 } = 0   }  \nonumber  \\
&& {}  \Rightarrow \; \; 
\bal^2 q - \frac{ \bal_b } { b_2^ 2 } =0  \; \; \Rightarrow \; \;  b_2 =   \frac{ \sqrt{ \bal_b} }{  \bal \sqrt{q} },  \hspace{2.5cm}  \label{b2-optimal}
\end{eqnarray}
\begin{eqnarray}
\lefteqn{  \frac{ \partial \btau^{-1} (q, b_1, b_2) } { \partial b_1} = 0   } \nonumber  \\    
&&   \Rightarrow \; \; 
 \bal^2 (  1 + b_2 ) q + \bal_b \big( 1-  \frac{  \bal_f } { b_1^ 2 }  + \frac{1}{ b_2} \big)= 0    \nonumber \\
&& \Rightarrow \;\;
 \frac{ \bal_b \bal_f } { b_1^ 2 } 
 = \bal^2 q + \bal \sqrt{ \bal_b q } + \bal_b + \bal \sqrt{ \bal_b q }   \nonumber \\
&&  \Rightarrow \; \;  \frac{ \bal_b \bal_f } { b_1^ 2 }  = \big( \bal \sqrt{q} + \sqrt{\bal_b} \big)^2   \nonumber \\
&&   \Rightarrow \; \; b_1 =\frac{\sqrt{ \bal_b \bal_f } } { \bal \sqrt{q} + \sqrt{\bal_b} } ,   \label{b1-optimal}    \\
 \lefteqn{ \frac{ \partial \btau^{-1} (q, b_1, b_2) } { \partial q} = 0  }   \nonumber  \\ 
&&  \Rightarrow \;\;   \bal^2  \big[ 2  \sqrt{\bal_b} + b_1 +  ( b_1 + \bal ) b_2 \big]  q \ln q  \nonumber   \\
&&     \; \; {} -     \bal^2  \big[ 2  \sqrt{\bal_b} + b_1 +  ( b_1 + \bal ) b_2 \big] q  ( \ln q +1 )    \nonumber   \\
&&   \;\; \,{}     -   \bal_b  \big[ b_1 + \bal_f b_1^{-1}   + ( b_1 + \bal ) b_2^{-1}   \big]   (  \ln q + 1) =0   \nonumber  \\
&& {} \Rightarrow \;\;  \bal^2 q   \big[ 2  \sqrt{\bal_b} + b_1 +  ( b_1 + \bal ) b_2 \big]  +   \bal_b   ( \ln q + 1 ) \nonumber   \\
&&   \qquad  \;\; {} \cdot   \big[ b_1 + \bal_f b_1^{-1}   + ( b_1 + \bal ) b_2^{-1}   \big]  =0   \nonumber   \\
&& {} \Rightarrow \;\;  \bal^2 q \big[ 2  \sqrt{\bal_b} + b_1 +   b_1 b_2 + \bal   b_2 \big]  +   \bal_b   ( \ln q + 1 )   \nonumber  \\
&&   \qquad  \;\; {} \cdot    \big[ b_1 + \frac{ \bal_f } {b_1 }   + \frac{ b_1 } {b_2} + \frac{ \bal } { b_2 }   \big]   =0.    
  \label{q-optimal}
\end{eqnarray} 
Substitution of (\ref{b2-optimal}) and (\ref{b1-optimal}) into (\ref{q-optimal}) and  some rearrangements produce a transcendental equation
\begin{eqnarray*}     
\lefteqn{  2 ( \bal \sqrt{q} )^3 + ( \bal + 2 \sqrt{\bal_b} + \sqrt{\bal_f } ) ( \bal \sqrt{q} )^2  } \\
&& \quad {} + \sqrt{\bal_b} ( \bal + \sqrt{\bal_f } )  \, \bal \sqrt{q}  
 + \big[ ( \bal + \sqrt{\bal_f } )  ( \bal \sqrt{q} )^2   \\
&& \quad {} + \sqrt{\bal_b} ( \bal + 3\sqrt{\bal_f } )  \, \bal \sqrt{q} + 2 \bal_b \sqrt{\bal_f } \, \big]  ( \ln q + 1  )   \\
&&  = (\bal \sqrt{q} + \sqrt{ \bal_b} ) \Big\{  2 ( \bal \sqrt{q} )^2 + (\bal + \sqrt{ \bal_f} ) \bal \sqrt{q}  \\
&& \quad  {} +  \big[ ( \bal + \sqrt{ \bal_f } ) \bal \sqrt{q} + 2 \sqrt{ \bal_b \bal_f} \, \big] ( \ln q +1 ) \Big\} =0,
\end{eqnarray*}
which is equivalent to equation (\ref{def-btau-prime-ce}) due to   $\bal \sqrt{q} + \sqrt{ \bal_b} >0$. 
It is observed from (\ref{def-btau-prime-ce})  that $\btau^{\prime} (\cdot)$ is   continuous and increasing  on $(0, \infty)$ as well as $\btau^{\prime} (e^{-1}) >0 $ and $ \btau^{\prime} (q) \to - \infty$ as $q \to 0$.  So  $\btau^{\prime} (\cdot)$ has a unique root $q_\ast \in (0, e^{-1})$ and $q_\ast  $ can be obtained   by solving (\ref{def-btau-prime-ce}) with initial guess $e^{-1}$. By   (\ref{b2-optimal}) and (\ref{b1-optimal}),
\begin{equation}  \label{b1b2-determined}
 b_1^\ast = \frac{\sqrt{ \bal_b \bal_f } } { \bal \sqrt{q_\ast} + \sqrt{\bal_b} } \quad {\rm and } \quad  b_2^\ast =   \frac{ \sqrt{ \bal_b} }{  \bal \sqrt{q_\ast} }.
\end{equation}
 The triple $(q_\ast, b_1^\ast, b_2^\ast)$ is the unique solution to the optimization problem (\ref{min_btau-1}) and gives the minimum $ \btau^{-1} (q_\ast, b_1^\ast, b_2^\ast)$. 
Setting   $b_1 =b_1^\ast$ and $b_2 =b_2^\ast$  in (\ref{LV_2-ce}) as well as (\ref{def-btau-ce}) produces 
$$
   \htau (q) = \btau (q, b_1^\ast, b_2^\ast) =  \btau (q, \frac{\sqrt{ \bal_b \bal_f } } { \bal \sqrt{q_\ast} + \sqrt{\bal_b} },  \frac{ \sqrt{ \bal_b} }{  \bal \sqrt{q_\ast} }) \quad \forall \, q \in (0,1)
$$
  which rearranges to   (\ref{def-htau-ce}).
From  (\ref{def-htau-ce}),  (\ref{def-btau-ce}) as well as  (\ref{def-btau-1-ce}), 
 $$
 \htau (q_\ast) = \btau  (q_\ast, b_1^\ast, b_2^\ast) = \frac{1}{  \btau^{-1} (q_\ast, b_1^\ast, b_2^\ast)}
$$
 is the maximum of   functions  (\ref{def-htau-ce}) as well as (\ref{def-btau-ce}).  So  (\ref{ImpulseInterval-ce}) implies that   condition (\ref{ImpulseInterval}) holds. By Theorem \ref{Theorem_bB},     CPS  (\ref{SiDE-bB}) and, hence, system (\ref{sampledSDE-bB})   are mean-square exponentially stable and  are also almost surely exponentially stable.
\end{proof}

\begin{remark}   \label{ImpulseInterval-ce-rewrite}
In Theorem \ref{Theorem-ce-1}, we show the mechanism of sampled-data system  (\ref{sampledSDE-bB}) by approach of controller emulation (process-oriented models) and  an innate relationship  (\ref{ImpulseInterval-ce})  between the control design  (\ref {bLV-bB})  and the sampling intervals of  implementation. One can let   $\barr = \bal \sqrt{ q}$ and rewrite condition (\ref{ImpulseInterval-ce}) as
\begin{equation}  \label{ImpulseInterval-ce-r}
  0 < \underline{\Delta t} \le  \overline{\Delta t} <  \htau (\barr_\ast)  
\end{equation}
to see what a key role the  control design  (\ref {bLV-bB}) plays in the sampled-data system,
where     $\htau : (0, \bal)  \to \RR_+$ is given as
\begin{eqnarray*}   
 \htau (\barr)   = -  2 \barr_\ast   \barr ^2  ( \ln \barr - \ln \bal ) \cdot    \Big\{ \sqrt{ \bal_b} \big[ \big( 2  \barr_\ast   + \bal +   \sqrt{\bal_f } \big)  \barr^2  \\
   {}  +  ( \bal  + \sqrt{\bal_f } ) \barr_\ast^2    + 2  \sqrt{  \bal_b \bal_f }  \barr_\ast  \big]  \Big\}^{-1}
\end{eqnarray*}
and $\barr_\ast = \bal \sqrt{ q_\ast } \in (0, \bal / \sqrt{e} )$ is the unique root of 
\begin{eqnarray*}   
  \btau^{\prime} (\barr) := 2 \barr^2 + (\bal + \sqrt{ \bal_f} ) \barr  +  2 \big[ ( \bal + \sqrt{ \bal_f } ) \barr + 2 \sqrt{ \bal_b \bal_f } \, \big] \\
 {} \cdot  \big[ \ln \barr - \ln ( \bal / \sqrt{e} ) \big]   =0 .
\end{eqnarray*}
\end{remark}

\begin{remark}   \label{ImpulseInterval-ce-partial}  
Substituting (\ref{b2-optimal}) and (\ref{b1-optimal}) into (\ref{def-btau-1-ce}), one can have
\begin{equation*}   
 \btau^{-1}  (q, \bal, \bal_b, \bal_f ) = - \frac{2 \sqrt{ \bal_b } }{ q \ln q }  \Big[ q + \big(1 + \frac{ \sqrt{ \bal_f} } { \bal } \big) \sqrt{q}  +   \frac{ \sqrt{ \bal_b \bal_f }} { \bal^2}    \Big]
\end{equation*}
for all $ 0 < q < 1$ and $\bal, \bal_b, \bal_f >0$, and observe 
that, given any $ q \in (0, 1) $,  function $\btau^{-1}$ is increasing with respect to either $\bal_b$ or $\bal_f$ while it is descreasing with respect to $\bal$.
\end{remark}

To disclose the equivalence and inherent relationship between the two main approaches,   we employ the same Lyapuov function $V(x) = \tV (x) = x^T P x$ for both the physical and the cyber subsystems  in  Theorem \ref{Theorem-ce-1}  as well as Theorem \ref{Theorem-dta}. Obviously, this could lead to conservative  results. Let us develop a  result for application using a couple of Lyapunov functions, which is  suggested in  Theorem \ref{Theorem_ExpStab} and Theorem \ref{Theorem_bB}.

\begin{theorem}   \label{Theorem-ce-2}
Suppose that the Lyapunov function $V(x)$  with condition (\ref {bLV-bB})  for physcial system  (\ref{controlledSDE}) is of the quadratic form (\ref{LyapunovFunction-quadratic}).
Let the sampling sequence $\{ t_k \}_{k \in \NN}$ satisfy 
\begin{equation}  \label{ImpulseInterval-ce-2}
  0 < \underline{\Delta t} \le  \overline{\Delta t} <  \htau (q_\ast)  
\end{equation}
where function $\htau : (0, 1) \to \RR_+$ is defined as 
\begin{equation}     \label{def-htau-ce-2}
\htau  (q ) = \frac{ - \bal^2 q \ln q } { \bal_b \gm_1 + \gm_2 \bal^2 q}
\end{equation}
and $q_\ast  \in (0, e^{-1} )$ is the unique root of 
\begin{equation}    \label{def-btau-prime-ce-2}
  \btau^{\prime} (q) :=  \bal^2 \gm_2 q + \bal_b \gm_1  (   \ln   q  +1  )  =0
\end{equation}
with     $\bal_b$, $\gm_1$ and $ \gm_2$ being positive numbers such that 
\begin{eqnarray}    
& V (\bB x) \le \bal_b  \tV(x) \quad \forall \,  x \in \RR^n ,  \label{cp-ce-1}  \\  
&  \tll  \tV (x, y) \le  \gm_1  V(x)  + \gm_2 \tV (y)  \quad \forall \,  x, y \in \RR^n     \label{cp-ce-2}    
\end{eqnarray}
for some qudratic function $\tV (x) = x^T \tP x$ defined by   $\tP >0$.  Then CPS (\ref{SiDE-bB})   is mean-square exponentially stable and  is also almost surely exponentially stable, which implies that its subsystem (\ref{SDE-bB_x}), viz.,   (\ref{sampledSDE-bB})    is mean-square exponentially stable and  is also almost surely exponentially stable.
\end{theorem}

\begin{proof}  According to Theorem \ref{Theorem_bB}, the assertion holds  if conditions  (\ref{LyapunovFunctions})-(\ref{ImpulseInterval})    with $p=2$ are satisfied for system (\ref{SiDE-bB}). 

Let  $\tV (y)=y^T \tP y $   of the quadratic form as (\ref{LyapunovFunction-quadratic}) for the cyber subsystem (\ref{SDE-bB_y}). So
$\lambda_m (P) |x|^2 \le V(x) \le \lambda_M (P) |x|^2$ and $\lambda_m (\tP) |y|^2 \le \tV(y) \le \lambda_M ( \tP) |y|^2$  for all $x, y \in \RR^n$;  
 i.e., condition (\ref{LyapunovFunctions}) holds with positives 
$p=2$, $c_1  =\lambda_m (P) \le c_2 = \lambda_M (P) $ and $\tc_1=\lambda_m (\tP) \le \tc_2 = \lambda_M (\tP) $.  
There is   $\bal_b  \in (0,  \lmd_M (P)  | \bB |^2 / \lmd_m (\tP) ]$ such that (\ref{cp-ce-1}) holds due to
\begin{equation*}
V ( \bB x )   \le    \lmd_M (P) | \bB|^2 | x|^2 \le   \frac{\lmd_M (P)  | \bB|^2  }{ \lmd_m (\tP) } \tV (x).
\end{equation*}
As above, by (\ref{bLV-bB})  and \cite[Lemma 1]{huang2009a},  for all $x, y \in \RR^n$,
\begin{multline*}    
   \ll V (x, y ) =   2 x^T P [ \barf (x) + \bB x - \bB y ]  + {\rm trace} \big[ \barg^T (x) P \barg (x)   \big]     \\
  \le - 2 \bal V(x) - 2x^T P \bB  y   
\le -   \bal    V(x) +  \frac{   \bal_b} {\bal}  \tV (  y) .
\end{multline*}
  Hence (\ref{LV_1}) holds with $\al_1 =   \bal   $ and $\al_2 =     \bal_b  / \bal $. Recall that both $\barf_u (x)$ and $\barg (x)$ satisfy the linear growth conditions, that is,  $ | \barf_u (x) |^2 \le 2 ( \bL + | \bB|^2) |x |^2$ and $| \barg (x) |^2 \le \bL |x|^2$.  Similarly, 
\begin{align*}       
 &  \tll \tV (x, y ) =   2 y^T \tP [ \barf (x) + \bB x - \bB y ]  + {\rm trace} \big[ \barg^T (x) P \barg (x)   \big]     \\
 & =   2y^T \tP \barf_u (x)   + {\rm trace} \big[ \barg^T (x) \tP \barg (x)   \big]   - 2 y^T \tP \bB y    \\
 & \le   \tV ( \barf_u (x) ) +   \tV ( y ) + \lmd_M (\tP) | \barg (x) |^2 +   \tV (y ) +   \tV ( \bB y)  \\
 &  \le  \lmd_M ( \tP )  ( | \barf_u (x) |^2 + | \barg (x) |^2 ) +  2  \tV (y ) + \lmd_M ( \tP )  | \bB|^2 y^2   \\
 & \le   \lmd_M ( \tP ) ( 3 \bL + 2 | \bB|^2) |x|^2  +  2  \tV (y ) + \lmd_M ( \tP )   | \bB|^2  y^2  \\
 & \le \frac{ ( 3 \bL + 2 | \bB|^2)  \lmd_M ( \tP )} { \lmd_m (P)} V(x) + \frac{ 2 \lmd_m ( \tP ) + | \bB|^2 \lmd_M ( \tP )  }{ \lmd_m ( \tP )} V(y)
\end{align*}
for all $x, y \in \RR^n$. This implies that there exist positive numbers  $\gm_1 \in (0, ( 3 \bL + 2 | \bB|^2)  \lmd_M ( \tP ) / \lmd_m (P)]$ and 
$ \gm_2 \in (0, 2 + | \bB|^2 \lmd_M ( \tP ) / \lmd_m ( \tP ) ]$ such that (\ref{cp-ce-2}) is satisfied, which is  condition (\ref{tLtV_1}) with $\tal_1 = \gm_1$ and $\tal_2 = \gm_2$.

Due to $ y(t_k)  =0 $ for all  $  k \in \NN$, nonngegatives  $\tilde{\bt}_1$,   $ \tilde{\bt}_2 $ and $ \tilde{\bt}_3$ can be chosen for  (\ref{iii-constants}) with  arbitrary small $ {\al_1}^{-1} \al_2    \tilde{\bt}_1 + \tilde{\bt}_2 + \tilde{\bt}_3  >0$. Conditions (\ref{EtV_tk-1})-(\ref{iii-constants})  hold. 

Substition of $\al_1 = \bal$, $\al_2 = \bal_b / \bal$, $\tal_1 = \gm_1$ and $\tal_2 = \gm_2$ into (\ref{ImpulseInterval}) and (\ref{def-hatq-ast}) produce (\ref{ImpulseInterval-ce-2}) and (\ref{def-btau-prime-ce-2}), respectively.  
Hence (\ref{ImpulseInterval-ce}) implies that condition (\ref{ImpulseInterval}) holds. 
By Theorem \ref{Theorem_bB},     systems  (\ref{SiDE-bB}) and, hence,   (\ref{sampledSDE-bB})   are mean-square exponentially stable and are also almost surely exponentially stable.
\end{proof}

\subsection{Discrete-time approximation (Computer-oriented models)}

As periodic sampling  ($ \{ t_k \}_{k \in \NN} $ with sampling period $ \Delta t = \underline{\Delta t} =  \overline{\Delta t} $)  is normally used \cite{astrom1997,nesic2001,nesic2006,zheng2002},  a sampling interval $t_k - t_{k-1}$  could vary  in the design method based on computer-oriented models which are  discrete-time approximation of the underlying continuous-time plants \cite{nesic2004, oishil2010}. 
By approach of discrete-time approximation,   one employs some approximate discrete-time model, say, the Euler-Maruyama approximation  of the continuous-time plant  (due to the usual unavailability of the exact discrete-time model),   and  designs a discrete-time  state-feedback controller $ \bu (X )  = \bB X $ for stability of the closed-loop system, which is the  Euler-Maruyama approximation \cite{higham2000,mao2007book,nesic2006} of the closed-loop system (\ref{controlledSDE}), 
\begin{equation}   
 X_k = X_{k-1} +    \barf_u (X_{k-1} )   h     + \barg( X_{k-1} )    \Delta B_k    \label{controlledSDE-EM}
\end{equation}
with  stepsize $ h >0$ and initial value $X_0 =x_0 \in \RR^n$, where  $\Delta B_k = B( k h) - B ( (k-1) h)$ for all $k \in \NN$. Specifically, a state-feedback controller $ \bu (X )  = \bB X $ is designed such that
\begin{equation}    \label{dV-bB}
 \EE [ V ( X_k ) | X_{k-1} ] \le (1 -\bc ) V ( X_{k-1})   \quad   \forall  \, X_{k-1} \in \RR^n
\end{equation}
and, therefore, the closed-loop system (\ref{controlledSDE-EM}) is exponentially stable \cite{boyd1994,huang_partI,khas2012}, where $\bc \in (0, 1)$ is a  constant and  $V: \RR^n \to \RR_+$ is a Lyapunov function  with (\ref{LyapunovFunctions_x}), say, the quadratic Lyapunov function  (\ref{LyapunovFunction-quadratic}).  
 The obtained controller $ \bu (x  )  = \bB x $  is then implemented in the continuous-time plant  using a sampler and ZOH device, that is,  $\bu (t ) = \bu (x (t_\ast ) ) = \bB x (t_\ast )   $ for all $t \ge 0$. This  leads to the sampled-data control system (\ref{sampledSDE-bB}) and its cyber-physical model  (\ref{SiDE-bB})  as well. The central question in the design method (\ref{dV-bB})
is, see \cite{astrom1997,nesic2001,nesic2004,nesic2006},   

 {\it  for what sampling sequence $ \{ t_k \}_{k \in \NN} $ does the sampled-data control system  (\ref{sampledSDE-bB}) share the stability property of the approximate discrete-time  model  (\ref{controlledSDE-EM})? }
 
We address this   question with Theorem \ref{Theorem-ce-1} and show the equivalence of the design methods   (\ref{bLV-bB}) and (\ref{dV-bB}).
\begin{theorem}   \label{Theorem-dta}
Suppose that the Lyapunov function $V(x)$  with condition (\ref {dV-bB})  for  cyber  system  (\ref{controlledSDE-EM}) is of the quadratic form  (\ref{LyapunovFunction-quadratic}).
Let the sampling sequence $\{ t_k \}_{k \in \NN}$ satisfy 
\begin{equation}  \label{ImpulseInterval-dta}
  0 < \underline{\Delta t} \le  \overline{\Delta t} <  \htau (r_\ast)  
\end{equation}
where function  $\htau : (0, \bal) \to \RR_+$  
 is given as
\begin{multline*}   
\htau (r)  = -  2 \, r_\ast   r^2  (  \ln r -  \ln \bal )
 \cdot    \Big\{ \sqrt{ \bal_b} \big[ \big( 2  r_\ast   + \bal +   \sqrt{\bal_f } \big)      r^2    \\
{}   +  ( \bal  + \sqrt{\bal_f } )  r_\ast^2    + 2  \sqrt{  \bal_b \bal_f }   r_\ast  \big]  \Big\}^{-1}
\end{multline*}
with $\bal = ( \bc h^{-1} + \bal_u h)  /2$ and $r_\ast =\bal \sqrt{q_\ast}  \in (0, \bal / \sqrt{e} )$
is the unique root of 
\begin{multline*}   
  \btau^{\prime} ( r) := 2 r^2 + (\bal + \sqrt{ \bal_f} ) r 
  +  2 \big[ ( \bal + \sqrt{ \bal_f } ) r + 2 \sqrt{ \bal_b \bal_f } \, \big]  \\
{} \cdot  \big[   \ln r -  \ln ( \bal / \sqrt{e} \, )  \big] =0.
\end{multline*}
with $\bal_u  $ being a positive constant such that
\begin{equation}    \label{drift-dta}
    V ( \barf_u (x) ) \le \bal_u V (x)  \quad \forall \, x \in \RR^n
\end{equation}
 as well as $\bal_b$ and $ \bal_f$  given by  (\ref{control-drift}).  Then CPS  (\ref{SiDE-bB})   is mean-square exponentially stable and  is also almost surely exponentially stable, which implies  that its subsystem (\ref{SDE-bB_x}), viz.,   (\ref{sampledSDE-bB})    is mean-square exponentially stable and  is also almost surely exponentially stable.
\end{theorem}

\begin{proof}
By  the design method (\ref {dV-bB})   as well as (\ref{LyapunovFunction-quadratic}) and (\ref{drift-dta}), 
\begin{eqnarray}
\lefteqn{  \EE [ V ( X_k ) | X_{k-1} ] = \EE \big[   X_k^T P X_k  | X_{k-1} \big]   }    \nonumber   \\
&& =\EE \Big[ \big( X_{k-1} +    \barf_u (X_{k-1} )   h     + \barg( X_{k-1} )    \Delta B_k \big)^T P  \nonumber  \\
&&  \;\;  {} \cdot \big( X_{k-1} +    \barf_u (X_{k-1} )   h     + \barg( X_{k-1} )    \Delta B_k \big) \big| X_{k-1}  \Big]  \nonumber  \\
&& = V ( X_{k-1} ) +  h \Big[ X_{k-1}^T P \barf_u ( X_{k-1} ) + \barf_u^T  ( X_{k-1} ) P X_{k-1}    \nonumber  \\
&&   \;\;  {} + {\rm trace} \big[ \barg^T (X_{k-1}) P \barg (X_{k-1} )    \big]   +  h  V( \barf_u ( X_{k-1} ) ) \Big]  \nonumber  \\
&& \le  V ( X_{k-1} ) +  h \Big[ 2 X_{k-1}^T P \barf_u ( X_{k-1} )   \nonumber  \\
&&    \;\;   {}  + {\rm trace} \big[ \barg^T (X_{k-1}) P \barg (X_{k-1})    \big]  + \bal_u  h  V ( X_{k-1} )\Big]     \nonumber   \\
&&   \le (1 -\bc ) V ( X_{k-1} )   \quad \forall \, X_{k-1} \in \RR^n         \label{dV-bB-inproof}
\end{eqnarray}
and, therefore, for all $X_{k-1} \in \RR^n$,
\begin{eqnarray}
 \lefteqn{  2 X_{k-1}^T P \barf_u ( X_{k-1} )   }  \nonumber  \\
  && \;\;  {}  + {\rm trace} \big[ \barg^T (X_{k-1}) P \barg (X_{k-1}) \big]  + \bal_u  h  V ( X_{k-1} )   \nonumber  \\
 &&  \le -   \frac{\bc}{ h } V ( X_{k-1} )     \label{dV-bLV-dta}
\end{eqnarray}
 where    $ \bal_u \in (0, 2 ( \bar{L} + |B|^2) \lmd_M (P) / \lmd_m (P) ]$   for (\ref{drift-dta})  due to
\begin{multline*}
V ( \barf_u (x)   ) \le \lmd_M (P) | \barf_u (x)  |^2       
\le 2 (  \bar{L} + |\bB   |^2)  \frac{ \lmd_M (P) }{ \lmd_m (P) } V (x) .
\end{multline*}

Let $V (x) = x^T P x$ also be the candidate Lyapunov function for continuous-time system (\ref{controlledSDE}). From (\ref{bLV}) and (\ref{dV-bLV-dta}), 
\begin{eqnarray}     \label{bLV--bB-dta}
   \ll V(x) =  2 x^T P \barf_u ( x)  + {\rm trace} \big[ \barg^T (x) P \barg (x) \big]   \nonumber  \\
  \le  - \big( \frac{ \bc }{h}  + \bal_u h  \big)   V (x)  \qquad \;\;   \forall  \, x \in \RR^n .
\end{eqnarray}
 This is exactly the design method (\ref{bLV-bB}) with  Lyapunov exponent, or say, decay rate
\begin{equation}  \label{bLV-dV-bB}
  2 \bal=   \frac{ \bc }{h}  + \bal_u h    . 
\end{equation}

On the other hand, if a controller is design for continuous-time system  (\ref{controlledSDE})  with  (\ref{bLV-bB}), by (\ref{dV-bLV-dta}), (\ref{bLV--bB-dta}) as well as  (\ref{dV-bB-inproof}),  one can choose  any stepsize $ h \in (0, (2 \bal / \bal_u ) \wedge ( 2 \bal )^{-1})$ and  then $\bc =  (2\bal - \bal_u h) h \in (0, 1)$ so that  condition  (\ref {dV-bB})  of  the other design method is satisfied. This with (\ref{bLV-dV-bB}) shows the equivalence of the design methods of  (\ref{bLV-bB}) and (\ref{dV-bB}). 

By (\ref{bLV-dV-bB}),  
let $r =\bal \sqrt{q} = ( \bc h^{-1} + \bal_u h) \sqrt{q} /2 $.  
Condition (\ref{ImpulseInterval-ce})  of Theorem \ref{Theorem-ce-1} can be written as 
(\ref{ImpulseInterval-dta}).  It   follows from Theorem \ref{Theorem-ce-1} that systems   (\ref{SiDE-bB}) and   (\ref{sampledSDE-bB}) are  mean-square exponentially stable and are also almost surely exponentially stable.
\end{proof}

\begin{remark} \label{stepsize-samplingperiod}
In the literature, periodic sampling is normally used and it is usually assumed that  the sampling period $  \Delta t $ is also the stepsize $h$ of the discrete-time model (i.e., $h=\Delta t$)
\cite{astrom1997,nesic2001,nesic2004,nesic2006, zheng2002}. They could be the same, namely, $h=\Delta t$ if the exact discrete-time model can be utilized, for instance, in   linear deterministic systems  \cite{astrom1997,seuret2012,zheng2002}.  
But, especially when some discrete-time approximation is employed (due to  unavailability of the exact   model),  the stepsize $h$ of the cyber model and the sampling period $ \Delta t $ are essentially different parameters  of the controller. The former is one of the design parameters   and the latter   a parameter of the implementation using a sampler and ZOH device. For stability of the resulting sampled-data control system (\ref{sampledSDE-bB}), we clearly show by  (\ref{ImpulseInterval-dta}) how the design parameters impose the maximum alllowable sampling interval on the implementation. 
\end{remark}

\begin{remark}  \label{equivalence-ce-dta}
We have shown the equivalence of the design methods   (\ref{bLV-bB}) and   (\ref{dV-bB}) for sampled-data control system (\ref{sampledSDE-bB}). Specifically, we not only provide the link \cite{seuret2012} but also reveal the intrinsic relationship  (\ref{bLV-dV-bB}) between the two main approaches. It is also observed that,  in addtion to  $P, \bal_b, \bal_f$ involved in   both  (\ref{bLV-bB}) and   (\ref{dV-bB}), a few  parameters $h, \bc, \bal_u$  are involved in the design method  (\ref{dV-bB}) as only one $\bal$ in the other.
\end{remark}


\section{Stability and stabilization of linear sampled-data systems}
As application of our established theory, 
we study stability and stabilization of linear sampled-data stochastic systems in this section. 
Let us consider    linear sampled-data control system
\begin{equation}    \label{linearSDE_sampled}
   \d x(t) = [ A x(t)  +  \bB x (t_\ast)  ] \d t+ \sum_{j=1}^m G_j x(t) \d B_j (t) \;\; t \ge 0
\end{equation}
  with initial value $x(0)=x_0 \in \RR^{n}$, where $A  \in \RR^{n \times n}$ and $ G_j \in \RR^{n \times n}$, $j=1, \cdots, m$, are constant matrices. The linear system (\ref{linearSDE}) is a specific case of  (\ref{sampledSDE-bB}) with $ \barf (x) = A x$ and $ \barg (x) = \begin{bmatrix} G_1 & \cdots & G_m  \end{bmatrix} x$, By Lemma \ref{existence_n_uniqueness}, it has a unique solution $x(t)$ on $ [0, \infty)$. It is well known that the continuous-time plant
\begin{equation}    \label{linearSDE}
   \d x(t) = F x(t)  \d t+ \sum_{j=1}^m G_j x(t) \d B_j (t) \quad t \ge 0
\end{equation}
 with  $F = A + \bB $ is mean-square exponentially stable
 if and only if there is a positive definite matrix $P \in \RR^{n \times n}$ such that
\begin{equation}    \label{Lyapunov-Ito-LMI}
     F^T P + P  F + \sum_{j=1}^m G_j^T P G_j \le - 2 \bal P   
\end{equation}
for some constant $\bal >0$. This is the  Lyapunov-It{\^o} inequality \cite{boyd1994},   the linear matrix inequality (LMI)  equivalent to the classical Lyapunov-It{\^o} equation  \cite{liu1992}.
By  \cite[Theorem 5.15, p175]{khas2012} or \cite[Theorem 4.2, p128]{mao2007book}, the mean-square exponential stability of SDE (\ref{linearSDE}) implies that it is also almost surely exponentially stable. Unlike  linear deterministic systems,  design methods  base on the exact discrete-time models \cite{astrom1997,oishil2010,seuret2012,zheng2002} are  not applicable to the stochastic system (\ref{linearSDE_sampled}).  Some  discrete-time approximation of the continuous-time plant has to be employed instead.
As a specific case of (\ref{controlledSDE-EM}),  the  Euler-Maruyama approximation  of  linear  system (\ref{linearSDE}) is
\begin{equation}   
 X_k = X_{k-1} +   F X_{k-1}    h     +  \sum_{j=1}^m G_j  X_{k-1}    \Delta B_{j, k}     \label{linearSDE-EM}
\end{equation}
with  stepsize $ h >0$ and initial value $X_0 =x_0 \in \RR^n$,  where $\Delta B_{j, k} = B_j ( k h ) -B_j ((k-1) h) $ for all $k \in \NN$.
It is also well-known that the discrete-time system (\ref{linearSDE-EM}) is mean-square
exponentially stable if and only if there exists a positive
definite matrix $P \in \RR^{n \times n}$ such that, see, e.g., \cite{boyd1994},
\begin{equation}    \label{Lyapunov-Ito-LMI-EM}
    ( I_n + h F )^T P   ( I_n + h F )  +  h  \sum_{j=1}^m G_j^T P G_j \le  (1 -\bc) P
\end{equation}
for some   $\bc \in (0, 1)$. Note that  (\ref{Lyapunov-Ito-LMI}) and (\ref{Lyapunov-Ito-LMI-EM})  are the specific cases of   the design methods   (\ref{bLV-bB}) and   (\ref{dV-bB}), respectively.  The equivalence of (\ref{Lyapunov-Ito-LMI}) and (\ref{Lyapunov-Ito-LMI-EM}) has shown by  the relationship (\ref{bLV-dV-bB}) for any stepsize $ h \in (0, (2 \bal / \bal_u ) \wedge ( 2 \bal )^{-1})$,  where $\bal_u >0$ is such that $ F^T P F \le \bal_u P$ in the linear system. The equivalence  of (\ref{Lyapunov-Ito-LMI}) and (\ref{Lyapunov-Ito-LMI-EM}) has also been addressed in \cite{huang_partI}.

Since we have shown the equivalence of the two main approaches  (\ref{bLV-bB}) and (\ref{dV-bB}),  let us  focus on  sampled-data control systems, say, by approach of controller emulation  (process-oriented models). A special version of Theorem \ref{Theorem-ce-2} for linear sampled-data stochstic   system (\ref{linearSDE_sampled}) is specified as follows. 

\begin{theorem}   \label{Theorem-ce-linear}
Suppose that there is a  positive definite matrix $P \in \RR^{n \times n}$ such that  LMI (\ref{Lyapunov-Ito-LMI}) holds for some    constant  $\bal >0$.
Let the sampling sequence $\{ t_k \}_{k \in \NN}$ satisfy (\ref{ImpulseInterval-ce-2}),  
where function  $\htau : (0, 1) \to \RR_+$ is defined by (\ref{def-htau-ce-2}) 
and $q_\ast  \in (0, e^{-1} )$ is the unique root of equation (\ref{def-btau-prime-ce-2}) 
with     $\bal_b$, $\gm_1$ and $ \gm_2$ being positive numbers such that 
\begin{eqnarray}    
& \bB^T P  \bB  \le \bal_b  \tP ,  \label{cp-ce-1-linear}  \\  
&  \begin{bmatrix}   \sum_{j=1}^m G_j^T \tP G_j  &  F^T \tP  \\  \tP F & - \bB^T \tP - \tP \bB   \end{bmatrix}  \le \begin{bmatrix}  \gm_1 P & 0 \\ 0 &  \gm_2 \tP \end{bmatrix}    \label{cp-ce-2-linear}    
\end{eqnarray}
for some positive definite matrix $\tP \in \RR^{n \times n}$.  Then   sampled-data control system  (\ref{linearSDE_sampled})   is mean-square exponentially stable and  is also almost surely exponentially stable.
\end{theorem}

Use $V(x) = x^T Px$ and $\tV(y) = y^T \tP y$ as the candidate Lyapunov functions for the physical and the cyber subsystems, respectively.  The LMIs (\ref{cp-ce-1-linear})-(\ref{cp-ce-2-linear}) imply the conditions (\ref{cp-ce-1})-(\ref{cp-ce-2}), respectively. Clearly, Theorem \ref{Theorem-ce-linear} is the direct application of Theorem \ref{Theorem-ce-2} to linear sampled-data stochstic   system (\ref{linearSDE_sampled}).

Letting  $\bB = \hB \hK$ with some given matrix $\hB  \in \RR^{n \times \hm}$ in system (\ref{linearSDE_sampled}) leads to the state-feedback stabilization problem of the    sampled-data  system,   which  requires to find a feedback gain matrix $\hK \in \RR^{\hm \times n}$ as well as some conditions on the sampling intervals for   stability of the 
closed-loop system 
\begin{equation}    \label{linearSDE_sampled-BK}
   \d x(t) = [ A x(t)  +  \hB \hK x (t_\ast)  ] \d t+ \sum_{j=1}^m G_j x(t) \d B_j (t)
\end{equation}
for all $t \ge 0$. It is reasonable in some sense to set $\tP =\tc P $ for some $\tc >0$   due to the interrelation of the the physical and the cyber subsystems in    CPS  (\ref{SiDE-bB}), see also \cite{huang_partI,liu1988}.  Applying  Theorem \ref{Theorem-ce-linear}, we obtain a useful result on feedback stabilization of   sampled-data  system (\ref{linearSDE_sampled-BK}),  which is  formulated as a set of LMIs with prescribed $\tc >0$,  see   \cite{fridman2004,huang2009a,nagh2008} as well.

\begin{theorem}   \label{Theorem-ce-linear-BK}
Suppose that there is a pair of matrices  $Q \in \RR^{n \times n}$ and  $Y \in \RR^{ \hm \times n}$ such that $ Q > 0$ and
\begin{equation}     \label{Lyapunov-Ito-LMI-BK}
\begin{bmatrix}   Q_{11}  + 2 \bal Q&    \ast  & \cdots  & \ast  \\                                  
                             G_1 Q  & - Q  & \cdots  & 0  \\
                   \vdots  & \vdots  & \ddots  & \vdots   \\
                            G_m Q  & 0  & \cdots  & - Q  \end{bmatrix}  \le 0
\end{equation}
 for some positive $\bal $, where  $Q_{11} = Q A^T + Y^T \hB^T + AQ + \hB Y$ and entries denoted by $\ast$ can be readily inferred from symmetry of a matrix.
Let the sampling sequence $\{ t_k \}_{k \in \NN}$ satisfy (\ref{ImpulseInterval-ce-2}),  
where function $\htau : (0, 1) \to \RR_+$ is defined by (\ref{def-htau-ce-2}) 
and $q_\ast  \in (0, e^{-1} )$ is the unique root of equation (\ref{def-btau-prime-ce-2}) 
with     $\bal_b$, $\gm_1$ and $  \gm_2 $ being positive numbers such that 
\begin{eqnarray}      
&   \begin{bmatrix}   - \bal_b   \tc\, Q &   \ast  
 \\  \hB Y & - Q  \end{bmatrix}    \le 0,  \label{cp-ce-1-linear-BK} \\
&   \begin{bmatrix}         - \gm_1 Q            &            \ast                          &    \ast  & \cdots  & \ast   \\      
                      \tc \, ( A Q + \hB Y )   &    \tQ_{22} - \gm_2 \tc Q  &   0  & \cdots  &   0   \\
                      \sqrt{ \tc} \, G_1 Q     &               0                          &   - Q    & \cdots  & 0  \\
                                \vdots              &          \vdots                        & \vdots & \ddots  & \vdots   \\
                      \sqrt{ \tc} \, G_m Q    &              0                           &      0    & \cdots  &  -  Q      \end{bmatrix}     \le 0  \quad
\label{cp-ce-2-linear-BK}
\end{eqnarray}
with $ \tQ_{22} = - \tc \,( Y^T \hB^T + \hB Y )$ for some prescribed number $\tc >0$.  Then the sampled-data control system  (\ref{linearSDE_sampled-BK}) with feedback gain matrix $\hK =Y Q^{-1}$  is mean-square  exponentially stable and   is also almost surely exponentially stable.
\end{theorem}

\begin{proof}
Let   $ P = Q^{-1}$ and $\tP = \tc  P$. Hence $P >0$ and $\tP >0$.  By the Schur complement lemma,  LMI (\ref{Lyapunov-Ito-LMI-BK}) produces
\begin{multline*}
  Q_{11}  + \sum_{j=1}^m QG_j^T P G_j Q   + 2 \bal Q \le 0      \quad  \Leftrightarrow  \\ 
   Q (A + \hB \hK )^T +   (A + \hB \hK ) Q +  \sum_{j=1}^m QG_j^T P G_j Q  \le - 2 \bal  Q . 
\end{multline*}
Premultiplying  by  $P$   and postmultiplying   by  $P$ the  LMI above gives the   LMI (\ref{Lyapunov-Ito-LMI})  with $ F = A + \hB \hK$.
By  the Schur complement lemma,  the LMIs (\ref{cp-ce-1-linear-BK}) and (\ref{cp-ce-2-linear-BK}) imply
\begin{eqnarray*}
 &       Q \hK^T \hB^T P \hB \hK Q - \bal_b \tc \, Q  \le 0 ,   \\
 &   \begin{bmatrix}  Q \sum_{j=1}^m G_j^T \tP G_j Q - \gm_1 Q & \ast    \\
       \tc \, ( A Q + \hB Y )    &        \tQ_{22} -  \gm_2 \tc \, Q  \end{bmatrix}  \le 0 .
\end{eqnarray*}
Premultiplying  by  $P$  and postmultiplying  by  $P$ the first one gives  (\ref{cp-ce-1-linear}) while premultiplying by  $ {\rm diag} \{ P, P \}$  and postmultiplying   by  $ {\rm diag} \{ P, P \}$ the second one yields   (\ref{cp-ce-2-linear}) with $\tP  = \tc \, P$. From Theorem \ref{Theorem-ce-linear}, the sampled-data control system  (\ref{linearSDE_sampled-BK}) with   $\hK =Y Q^{-1}$  is mean-square exponentially stable and  is also almost surely exponentially stable.
\end{proof}

\begin{remark}   \label{Remark-control-design}
  As an implementation of Theorem   \ref{Theorem-ce-linear-BK},  we propose an algorithm in the form of  generalized eigenvalue problems and LMIs \cite{boyd1994,boyd2004,gahinet1995}, which  finds a feasible solution to the set of LMIs (\ref{Lyapunov-Ito-LMI-BK})-(\ref{cp-ce-2-linear-BK}).  Assume  $m=1$ and $G_1 = G$ for simplicity.  
\begin{itemize}
\item[1)] Compute the maximum Lyapunov exponent   $1 / \lmd$ by solving  the  generalized eigenvalue minimization problem 
\begin{equation*}   
   \min \, \lmd \;\; s.t. \; \bQ  >0, \; \begin{bmatrix}  \bQ & 0    \\
                          0  & 0     \end{bmatrix}   <  \lmd  
\begin{bmatrix}  - \bQ_{11}  &    \ast  \\                                  
                            -G  \bQ  & \bQ    \end{bmatrix} 
\end{equation*}
 with  $\bQ_{11} = \bQ A^T + \bY^T \hB^T + A \bQ + \hB \bY$. 

\item[2)] Choose Lyapunov exponent    $2 \bal <1 / \lmd$ and obtain matrices $Q> 0$ and $Y$ by solving the LMI (\ref{Lyapunov-Ito-LMI-BK}). 

\item[3)] Find $ \bal_b$ by solving the LMI  (\ref{cp-ce-1-linear-BK})
with   $Q>0$ and $Y$ obtained in the previous step as well as  prescribed $\tc >0$.

\item[4)] Find $\gm_1 $ and $ \gm_2 $  by solving   the LMI (\ref{cp-ce-2-linear-BK})
with  $Q>0$ and $Y$ obtained in  step 2) as well as  prescribed $\tc >0$.
\end{itemize}
The obtained matrices $Q>0$, $Y$ and $\begin{bmatrix} \bal & \bal_b & \gm_1 & \gm_2 & \tc \end{bmatrix}$  not only produce  a feasible solution   to  the set of LMIs (\ref{Lyapunov-Ito-LMI-BK})-(\ref{cp-ce-2-linear-BK}) and   the state-feedback stabilization problem of sampled-data system (\ref{linearSDE_sampled-BK}) but also provide  starting points to find some other feasible solutions with larger allowlable sampling intervals (\ref{ImpulseInterval-ce-2}) using  toolboxes such as \cite{gahinet1995,gotb2011}. For a linear deterministic system (viz. system (\ref{linearSDE_sampled-BK}) with $G =0$),    $\tc $ can be, instead of   a prescribed number,  one of  the  decision variables $ \bal_c = \bal_b \tc >0$, $ \gm_c = \gm_2 \tc >0$ and $\tc >0$ in the LMIs  (\ref{cp-ce-1-linear-BK})-(\ref{cp-ce-2-linear-BK}), solving which gives positives   $\bal_b =   \bal_c  / \tc$,  $ \gm_2  =  \gm_c /  \tc $ and $ \tc$.
Notice that
 our  control design method can be applied  with Theorem \ref{Theorem-ce-2}   to nonlinear systems as well, see Example 2 below. 
\end{remark}

\section{Illustrative examples}

In this section, we illustrate the application of our proposed results with numerical examples in the literautre.

{\bf Example 1.}  Stabilization of stochastic systems by sampled-data control has been studied in quite a few works. Here we consider  two specific   cases of linear sampled-data stochastic system (\ref{linearSDE_sampled})
with  $m=1$. 
In  one case,
\begin{equation}   \label{ex1-sub1}
 A  = \begin{bmatrix} 1 & -1 \\  1 & -5 \end{bmatrix},  \;
 G = \begin{bmatrix} 1 &  1 \\  1 & -1 \end{bmatrix},  \;
 \bB= \begin{bmatrix} -10 & 0 \\ 0 & 0\end{bmatrix} ,  
\end{equation}
 and in the other,
\begin{equation}   \label{ex1-sub2}
 A  = \begin{bmatrix} -5 & -1 \\  1 & 1 \end{bmatrix},  \;
 G = \begin{bmatrix} -1 &  -1 \\  -1 & 1 \end{bmatrix},   \;
  \bB= \begin{bmatrix}  0 & 0 \\  0 & -10 \end{bmatrix} .
\end{equation}

\begin{center}
\begin{figure}[!t]

\vspace{-0.2cm}

\hspace{-0.9cm}
\includegraphics[width=10.5cm]{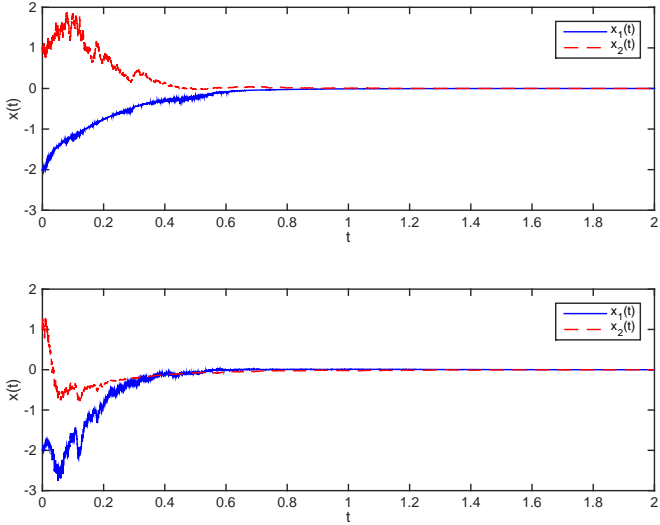}  

\vspace{-0.5cm}

 \caption{ 
A trajectory sample of   system  (\ref{ex1-sub1-control}) with  $\hK = [ -5.5085  \; \, -0.1520 ]$  (above) and that of system  (\ref{ex1-sub2-control})   with  $\hK =[ 0.1738 \; \,  -5.5639]$ (below). \label{Ex1_traj} 
}
\end{figure}
 \end{center} 

\vspace{-0.8cm}

Sampled-data stochastic systems (\ref{ex1-sub1}) and (\ref{ex1-sub2}) with sampling period $\tau >0$ have been studied in \cite{mao2013,mao2014,you2015}.  It is observed in \cite[Example 6.1]{you2015} that, by   \cite[Corollary 5.4]{you2015} with $N=1$, $Q = I_2$, $K_1= 5.236$, $K_2 = \sqrt{2}$, $K_3 =10$, $c_1 = c_2 = \lmd_1=1$,   $\lmd_2 =4$ and $\lmd_3=8$,  both the sampled-data systems (\ref{ex1-sub1}) and  (\ref{ex1-sub2}) are mean-square exponentially stable and   also almost surely exponentially stable if the sampling period 
$ \tau < \tau^\ast = 0.0074$, a better bound than  those  in \cite{mao2013,mao2014}. 

Let us apply Theorem \ref{Theorem-ce-linear} to  sampled-data stochastic systems  (\ref{ex1-sub1}) and (\ref{ex1-sub2}), respectively. 
For system  (\ref{ex1-sub1}), LMIs (\ref{Lyapunov-Ito-LMI}), (\ref{cp-ce-1-linear}) and (\ref{cp-ce-2-linear})  are satisfied with $\bal =4.3957$,  $\bal_b= 241.9335$, $\gm_1= 1.2491$, $\gm_2 = 60.5024$,  $P = \begin{bmatrix}   2.2173   &  0.8212 \\  0.8212 & 6.1228  \end{bmatrix}$ and 
$ \tP = \begin{bmatrix}   0.9193  &  -0.0046 \\  -0.0046 &   0.0178  \end{bmatrix} $. According to Theorem \ref{Theorem-ce-linear}, sampled-data system  (\ref{ex1-sub1}) is mean-square exponentially stable and is also almost surely exponentially stable if
\begin{equation*}
 0 < \uDt \le \oDt < \htau (q_\ast) = 0.0116. 
\end{equation*}
Similarly, for system  (\ref{ex1-sub2}), the LMIs 
 are satisfied with $\bal =4.4352$,  $\bal_b =6.5438$, $\gm_1= 57.5429$, $\gm_2=  61.6297$,  $P = \begin{bmatrix}   73.4547 &  -2.3459 \\  -2.3459 & 14.5076  \end{bmatrix}$ and 
$ \tP = \begin{bmatrix}58.3763  &9.0426 \\ 9.0426 &  240.5279  \end{bmatrix} $. It immediately follows from Theorem \ref{Theorem-ce-linear} that sampled-data system  (\ref{ex1-sub2}) is mean-square exponentially stable and is also almost surely exponentially stable  if 
\begin{equation*}
 0 < \uDt \le \oDt < \htau (q_\ast) = 0.0102. 
\end{equation*}
Our method has improved the existing results.

\begin{center}
\begin{figure}[!t]

\vspace{-0.75cm}

\includegraphics[width=8.5cm]{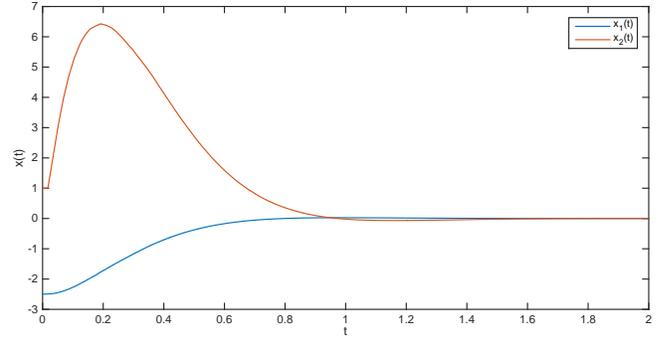}  

\vspace{-1cm}

 \caption{ The trajectory   of   system  (\ref{ex2}) with  $\hK = [ -27.5776 \;\,   -8.2817]$. \label{Ex2_traj} 
}
\end{figure}
 \end{center} 

\vspace{-0.8cm}

 Furthermore, as application of Theorem   \ref{Theorem-ce-linear-BK} and the control design method in Remark \ref{Remark-control-design}, we study the state-feedback stabilization problems of sampled-data system (\ref{linearSDE_sampled-BK}) with 
\begin{eqnarray}   
& A  = \begin{bmatrix} 1 & -1 \\  1 & -5 \end{bmatrix},  \;
 G = \begin{bmatrix} 1 &  1 \\  1 & -1 \end{bmatrix},  \;
 \hB= \begin{bmatrix} 1 \\  0\end{bmatrix} 
    \label{ex1-sub1-control}    \\
& {} {\rm and} \;\; A  = \begin{bmatrix} -5 & -1 \\  1 & 1 \end{bmatrix},  \;
 G = \begin{bmatrix} -1 &  -1 \\  -1 & 1 \end{bmatrix},   \;
  \hB= \begin{bmatrix}    0 \\ 1 \end{bmatrix} ,  \label{ex1-sub2-control}
\end{eqnarray}
respectively, see \cite{mao2013,mao2014,you2015} as well as \cite{huang2009a}.

For  system (\ref{ex1-sub1-control}),   the set of LMIs (\ref{Lyapunov-Ito-LMI-BK})-(\ref{cp-ce-2-linear-BK}) is satisfied with
$ \bal =3.6536,\bal_b =4.2422, \gm_1 =26.2456, \gm_2 =26.7130, \tc =7.2691$, $Q = \begin{bmatrix} 0.2593 &  0.0249   \\
    0.0249  &  0.2449  \end{bmatrix}$ and $Y = [ -1.4322   \;  -0.1744]$, which yields  feedback gain $\hK = Y Q^{-1}=[ -5.5085  \;\,  -0.1520 ]$ with $| \hK | = 5.5106 < 10$ smaller than the one in \cite{mao2013,mao2014,you2015}. But, by Theorem   \ref{Theorem-ce-linear-BK}, sampled-data control system (\ref{ex1-sub1-control}) with feedback gain matrix $\hK = [ -5.5085  \; \, -0.1520 ]$  is mean-square and almost surely exponentially stable if   the sampling intervals satisfy
\begin{equation}  \label{ex1-sampling}
 0 < \uDt \le \oDt < \htau (q_\ast) = 0.0235,
\end{equation} 
which is much larger than the bound $ \tau^\ast = 0.0074$ in \cite{you2015}.

For system (\ref{ex1-sub2-control}),  the
 LMIs (\ref{Lyapunov-Ito-LMI-BK})-(\ref{cp-ce-2-linear-BK}) hold with
$ \bal =3.7157$, $\bal_b =5.7100$, $\gm_1 =18.8231$, $\gm_2 =29.6417$, $ \tc =5.5547$, $Q = \begin{bmatrix} 194.7706 & -20.0691 \\  -20.0691 &  207.2345 \end{bmatrix}$ and $Y = [ 0.1455 \; \, -1.1565] \times 10^3$, which, by  Theorem   \ref{Theorem-ce-linear-BK},  implies  both the mean-square exponential stability and the almost sure exponential stability of the sampled-data control system  (\ref{ex1-sub2-control}) with feedback gain matrix $\hK = Y Q^{-1}=[ 0.1738 \; \, -5.5639]$. This  produces not only  smaller  gain  $| \hK | = 5.5667 < 10$ but also much larger allowable sampling intervals (\ref{ex1-sampling}) as well. 

Our   design method has improved the existing results significantly. Trajectory samples of the closed-loop  systems  (\ref{ex1-sub1-control}) and  (\ref{ex1-sub2-control}) with sampling period $\Delta t =0.0234< \htau (q_\ast) = 0.0235$ are shown in Figure \ref{Ex1_traj}, where   $x(0) =x_0= [ -2 \; \,1]^T$ cf. \cite{,mao2014,you2015}.

{\bf Example 2.} Let us illustrate application of our design method  to nonlinear systems   with a planar system 
 \cite{tsinias1991,qian2002}
\begin{equation*}
 \begin{bmatrix}  \dot{x}_1 \\ \dot{x}_2 \end{bmatrix} = 
\begin{bmatrix}  x_2 +\frac{1}{4} ( x_1 + x_1 \sin (u x_2)) \\ u + x_1 \sin (u x_2)  \end{bmatrix},
\end{equation*}
where $x = [x_1 \; \, x_2]^T \in \RR^2$ and $u \in \RR$ are the system state and input, respectively. It has been shown in \cite{tsinias1991} that the system can be globally stabilized by a linear state-feedback law $u = \hK x$ with some  gain  matrix  $\hK \in \RR^{1 \times 2}$.  The implemention of such a controller using a sampler and
ZOH device  leads to  a  specific case of sampled-data control system  (\ref{sampledSDE-bB}) in which  
\begin{multline}
 \barf (x) =   \bA x + \phi(x) ,  \;\;   \bB=\hB \hK,  \;\;  \barg (x) \equiv 0 ,   \;\;  \forall  x \in \RR^2    \\
 \bA = \begin{bmatrix}  \frac{1}{4}  & 1   \\   0 &  0   \end{bmatrix}, \;  \phi(x) =  \begin{bmatrix}   \frac{1}{4}  x_1 \sin (\hK x  x_2)  \\ x_1 \sin ( \hK  x x_2)  \end{bmatrix},  \;
\hB =  \begin{bmatrix}    0 \\ 1      \end{bmatrix}.   \label{ex2}
\end{multline}
System (\ref{ex2})  satisfies the local Lipschitz condition and the linear growth condition since, given   matrix $Q>0$,
 \begin{align*}
  &    \phi^T  (x) Q \phi (x)   =   \begin{bmatrix}   \frac{1}{4} & 1 \end{bmatrix} Q \begin{bmatrix}   \frac{1}{4} & 1 \end{bmatrix}^T x_1^2  \sin^2 (\hK x  x_2)      \\
  & \le \begin{bmatrix}   \frac{1}{4} & 1 \end{bmatrix} Q \begin{bmatrix}   \frac{1}{4} & 1 \end{bmatrix}^T x_1^2    =   x^T E_1^T Q E_1 x  \le \lmd_M (Q) | E_1|^2 |x|^2
\end{align*}
for all $x \in \RR^2$, where $E_1 =   \begin{bmatrix} \frac{1}{4}   & 0 \\  1  & 0 \end{bmatrix} $.  Given $V(x) $ and $ \tV (x) $ as Theorem \ref{Theorem-ce-2}, the conditions  (\ref {bLV-bB}), (\ref{cp-ce-1}), (\ref{cp-ce-2}) are specified as  a set of LMIs as follows
\begin{align*}
 &  \dot{V} (x) =   x^T (  \tA^T P +  P \tA    ) x + 2 x^T P \phi (x) \\
 & \le x^T   (    \tA^T P +  P \tA  + b P   ) x  +  b^{-1} \phi^T (x) P \phi (x) \\
 &  \le x^T  ( \tA^T P +  P \tA  + b P  + b^{-1} E_1^T P E_1 ) x \\
 &  \le -2 \bal V(x)  \\
 & \Rightarrow \quad   \tA^T P + P \tA   + b P  + b^{-1} E_1^T P E_1 \le - 2\bal P, \\
 & \hspace{4.55cm} {}  \bB^T P \bB   
     \le \bal_b \tP, \\
 & \;\; \begin{bmatrix} c^{-1} E_1^T \tP E_1  &  \tA^T \tP \\
   \tP \tA  &  - \bB^T \tP - \tP \bB + c \tP \end{bmatrix}  \le \begin{bmatrix}  \gm_1 P & 0 \\ 0 &  \gm_2 \tP \end{bmatrix} ,
\end{align*} 
where $\tA = \bA + \bB$ and both $b, c$ are  positive numbers.

Applying  our  control design method presented in   Remark \ref{Remark-control-design} with the set of LMIs aoove,  we   obtain  state-feedback gain matrix $\hK = [ -27.5776 \;\,   -8.2817] $, which, therefore, gives 
$ \bB  = \hB \hK 
= \begin{bmatrix}  0 & 0 \\ -27.5776 & -8.2817 \end{bmatrix}$ and $\tA  = \bA + \bB 
= \begin{bmatrix} 0.25 &   1 \\
  -27.5776  &  -8.2817 \end{bmatrix}$.
 The set of LMIs   is satisfied with $\bal =3.4369, \bal_b=   0.1507, \gm_1 =  137.2912 , \gm_2 =  142.0755, b=    0.4632,    c=37.5579$, $P = \begin{bmatrix} 3.0050 &   0.4509 \\  0.4509  &  0.0983 \end{bmatrix}$ and $\tP = \begin{bmatrix} 667.5859 & 161.7904 \\ 161.7904  & 45.8086 \end{bmatrix} $. By Theorem  \ref{Theorem-ce-2},   sampled-data control system (\ref{ex2}) with feedback gain $\hK = [ -27.5776 \;\,   -8.2817] $  is mean-square exponentially stable and is also almost surely exponentially stable if the sampling intervals satisfy 
$$ 0 < \uDt \le \oDt < \htau (q_\ast) = 0.0175.
$$
   The trajectory of  the controlled system  (\ref{ex2}) is shown in Figure \ref{Ex2_traj}, where   sampling period  $\Delta t =0.0174< \htau (q_\ast) = 0.0175$ and  initial value $x(0) =x_0= [ -2.5 \;\, 1]^T$.

\section{Conclusion and future work}

In this paper, we have presented  the cyber-physical model of a computer-mediated control system, which not only provides a holistic view but also reveals the inherent relationship between the physical system and the cyber system. Such cyber-physical dynamics can be expressed by our canonic form (\ref{SiDE-xy}) of CPSs, which is an extension of    \cite[Eq.(2.1)]{huang_partI} for synthesis of CPSs. We have established a Lyapunov stability theory for the synthetic CPSs  and applied it to stability analysis and feedback stabilization of computer-mediated control systems, which are typically  known as  sampled-data control systems.
This paper has contructed a foundational theory of computer-mediated control systems.

Our CPS theory can be further developed by many techniques of  Lyapunov functions/functionals \cite{fridman2010,nagh2008,huang202x}   such as constructing a Lyapunov function/functional for the whole CPS that could  improve our results   by exploiting the structure of  the composition of the subsystems \cite{liu1988,liu1992}. 
As application of our theory to sampled-data control systems, 
 we have addressed the keys questions in  two main approaches and revealed their equivalence and intrinsic relationship.
We have not only developed stability criteria but also proposed control design methods for state-feedback stabilization of sampled-data  systems.  In practice, feedback control  is usually based on an observer that is designed to reconstruct  the state using measurements of the input and the output of the system \cite{grizzle2007,nikoukhah1998,qian2002,rios2016}.  Our canonic form  (\ref{SiDE-xy}) of synthetic CPSs is able to include the dynamics of    observers as well as  impluse effects such as those in a robot model \cite{grizzle2007}. This is important for nonlinear control systems in which the so-called separation principle  may not hold  \cite{khalil2002,qian2002}.

 In this paper, we  have laid a theoretic foundation for computer-mediated control systems and initiated a system science for
CPSs.   This   arouses many interesting and challenging problems. For example, one can naturally generalize the time-triggered mechanism in CPS (\ref{SiDE-xy}) to an event-triggered mechanism \cite{huang202x} and the SiDE to a stochastic impulsive differential-algebraic equation (SiDAE) \cite{huang2011}  so that the CPS  can  encompass event-triggered sampling/control \cite{heemels2012,huang202x,tanwani2016} and equality constraints \cite{huang2011,nikoukhah1998} on both the physical and the cyber sides.  As an example, one of such  generalizations of synthetic CPS (\ref{SiDE-xy}) can be as follows
\begin{subequations}  \label{SiDAE-xy}
\begin{align}
&E_x \mathrm{d}  x(t) = f(x(t), y(t), t)\mathrm{d}t + g(x(t), y(t), t )\mathrm{d} B(t)    \label{SDAE_x} \\
& \hspace{4.8cm} {} t \in [ 0, \infty) \setminus \{t_k\}_{k \in \NN}  \nonumber   \\
&E_y \d y(t)  = \tilde{f} (x(t), y(t), t) \d t + \tilde{g}( x(t), y(t), t) \d B (t)   \label{SDAE_y}  \\
& \hspace{4.8cm} {} t \in [ 0, \infty) \setminus \{t_k\}_{k \in \NN}  \nonumber   \\
& \Delta (x_{t_k^-}, y_{t_k^-}, k) := x(t_k) -x(t_k^-)  \nonumber  \\
 & \; \;  {}   =\left\{  \begin{array}{cc}  h  (x_{t_k^-}, y_{t_k^-},   \bar{ \xi } (k ),  k), &   \kappa_x (x_{t_k^-}, y_{t_k^-}, k) > 0  \\
      0, & \kappa_x (x_{t_k^-}, y_{t_k^-}, k)  \le 0  \end{array}  \right.   \label{event-Impulse_x}   \\
&\tilde{\Delta} (x_{t_k^-}, y_{t_k^-},  k ) := y(t_k) - y(t_k^-)  \nonumber  \\
 & \; \; {}       = \left\{  \begin{array}{cc}  \tilde{h}  (x_{t_k^-}, y_{t_k^-},    \bar{ \xi } (k ),  k), &   \kappa_y (x_{t_k^-}, y_{t_k^-}, k) > 0  \\
      0, & \kappa_y (x_{t_k^-}, y_{t_k^-}, k)  \le 0  \end{array}  \right.   \label{event-Impulse_y} 
\end{align}
\end{subequations}
for all $k \in \NN$, where $E_x \in \RR^{n \times n}$ and $E_y \in \RR^{ q \times q}$ are constant matrices with $0< rank (E_x) \le n$ and $0< rank (E_y) \le q$, respectively;  $h : C  ([t_{k-1}, t_k); \RR^n) \times C  ([t_{k-1}, t_k); \RR^q)  \times \RR^n \times \NN \to \RR^n $,  $\tilde{h}  : C  ([t_{k-1}, t_k);  \RR^n) \times C  ([t_{k-1}, t_k); \RR^q)  \times \RR^n \times \NN  \to \RR^q $, $ \kappa_x: C  ([t_{k-1}, t_k);  \RR^n) \times C  ([t_{k-1}, t_k); \RR^q)   \times \NN \to \RR$ and $ \kappa_y: C  ([t_{k-1}, t_k);  \RR^n) \times C  ([t_{k-1}, t_k); \RR^q)   \times \NN \to \RR$ are measurable functions. 
Clearly,  the generalization (\ref{SiDAE-xy}) of CPSs has a much wider range of applications since differential-algebraic
equations   describe a great many natural
phenomena and event-triggered mechanisms of sampling/control are  increasingly popular in wired and wireless networked control systems  \cite{heemels2012,huang2011,huang202x,tanwani2016}.
 Our CPS  theory can be extended to various dynamical systems such as stochastic hybrid  systems \cite{teel2014} including stochastic systems with time delay, impulses as well as switching \cite{huang2009,huang2009a,huang202x} and  distributed parameter systems \cite{curtain1995,krstic2008}, in which stochastic stabilization \cite{higham2000,huang2015,mao2007book} is one of the many interesting topics. 
Moreover, the proposed CPS theory may be adapted to special control systems such as control systems
with actuator saturation \cite{fridman2004}, sliding mode control systems \cite{huang2010PhD}, sampled-data systems with controlled sampling as well as control systems with stabilizing delay \cite{seuret2012}.  It  is also of theoretic and practical importance to study  a CPS that  involves multi-scale processes in either or both of the  physical and the cyber sides \cite{huang2015armax,huang2016}, which could be a challenge. Just name a few  among   future work  to develop the systems science for CPSs. 

\section*{Acknowledgement}
The author is supported  by  National Natural Science Foundation of China (No.61877012).
The author    would like to thank Prof. D. Casta{\~n}{\'o}n   and the anonymous A.E. for their   careful reading and helpful comments on a previous version of \cite{huang_partI}.
\section*{Appendix}


  {\it Proof of Lemma \ref{existence_n_uniqueness}: }
Since system (\ref{Compact-z}) satisfies the locbal Lipschitz   condition (\ref{localLipschitz-z}) and linear growth condition (\ref{lineargrowth-z}), according to \cite[Theorem 3.4, p56]{mao2007book}, there exists a unique solution $z(t) = z(t; z_0)$  to SiDE (\ref{Compact-z}) on $t \in [t_0, t_1)$ and the solution belongs to $\mm^2 ([t_0, t_1); \RR^{n+q})$. Notice that $\xi (1)$ is $\ff_{t_1}$-measurable and independent of $\{ z(t): t  \in [t_0, t_1) \}$ while $H_F ( z_{t_1^-}, 1)$  and $\bH_G ( z_{t_1^-}, 1)$ are all $\ff_{t_1^-}$-measurable. By virtue of the continuity of functions $H_F ( \cdot, k)$  and $\bH_G (\cdot, k)$   with respect to their first arguments for all $k \in \NN$, there exists a unique solution $z(t_1)$ to  (\ref{Compact-z}) at $t = t_1$. Moreover, (\ref{Impulse_z}) and (\ref{localLipschitz-z}) imply that the second moment of $z(t_1)$ is finite. And, again, according to \cite[Theorem 3.4, p56]{mao2007book}, one has that there is a unique right-continuous solution $z(t)$  to  (\ref{Compact-z}) on $t \in [t_0, t_2)$ and the solution belongs to $\mm^2 ([t_0, t]; \RR^{n+q})$ for $t \in [t_0, t_2)$. Recall that $\{ t_k \}_{k \in \NN}$ with $t_1 > t_0 :=0$ is a strictly increasing sequence such that  $0<\underline{\Delta t} :=\inf_{k \in \NN} \{ t_k - t_{k-1} \} \le \overline{\Delta t} :=\sup_{k \in \NN} \{ t_k -t_{k-1} \} < \infty$  and hence $t_k \to \infty$ as $k \to \infty$. By induction, one has that there exists a unique (right-continuous) solution $z(t)$ to SiDE (\ref{Compact-z}) and the solution belongs to  $\mm^2 ([0, T]; \RR^{n+q})$ for all $T \ge t \ge 0$. Moreover,   according to \cite[Theorem 4.3, p61]{mao2007book}, $x(t)$ is continuous on each $t_k$ and hence on $t \in [0, T]$ for all $T \ge 0$ since (\ref{localLipschitz-x}) implies that  subsystem (\ref{SDE_x}) satisfies the linear growth conditon with respect to $x$ on each $t_k$ and $k \in \NN$. \hfill $\Box$

\vfill


\begin{thebibliography}{100}



\bibitem{astrom1997}
K.~J. {\r A}str{\"o}m,  B.~Wittenmark, \emph{Computer-controlled systems: theory and design (3rd Ed.)},  New Jersey, US: Prentice Hall, 1997.


\bibitem{astrom2014}
K.~J. {\r A}str{\"o}m,  P.~R. Kumar,  
``Control: a perspective,"
  \emph{Automatica}, vol. 50,  pp. 3-43, 2014.


\bibitem{boyd1994}
S.~Boyd, L.~El Ghaoui, E.~Feron, V.~Balakrishnan,   
  \emph{Linear matrix inequalities in systems and control theory},
 Pennsylvania, US: Society for Industrial and Applied Mathematics, 1994. 

\bibitem{boyd2004}
S.~Boyd, L.~Vandenberghe,   
  \emph{Convex Optimization},
 Cambridge, UK: Cambridge University Press, 2004. 


\bibitem{curtain1995}
R.~F. Curtain,  H.~Zwart, \emph{An introduction to infinite-dimensional linear systems theory},  New York, US: Springer-Verlag, 1995.


\bibitem{fridman2004}
 E.~Fridman,  A.~Seuret, J.-P. Richard,    ``Robust sampled-data stabilization of linear systems:
an input delay approach,"
  \emph{Automatica}, vol.40,  pp.1441-1446, 2004.

\bibitem{fridman2008}
 E.~Fridman,  M.~Dambrine, N. Yeganefar,  ``On input-to-state stability of systems with time-delay: a matrix
inequalities approach,"
  \emph{Automatica}, vol.44,  pp.2364-2369, 2008.

\bibitem{fridman2010}
 E.~Fridman,    ``A refined input delay approach to sampled-data control,"
  \emph{Automatica}, vol.46,  pp.421-427, 2010.

\bibitem{gahinet1995}
P.~Gahinet, A.~Nemirovski, A.~J. Laub, M.~Chilali, \emph{LMI control toolbox}, 
  Massachusetts, USA: The MathWorks Inc, 1995.


\bibitem{gotb2011}
 \emph{Global Optimization Toolbox User's Guide}, 
  Massachusetts, USA: The MathWorks Inc, 2011.





\bibitem{grizzle2007}
J.~W. Grizzle, J.~H. Choi, H.~Hammouri, B.~Morris, “On observer-based
feedback stabilization of periodic orbits in bipedal locomotion,”
\emph{ Meth.  Mod.  Autom.   Robot.}, August,
pp. 27-30, 2007.



 \bibitem{heemels2012}
W.P.M.H.~Heemels, K.H.~Johansson, P.~Tabuada,    ``An introduction to event-triggered and self-triggered control," 
in \emph{Proc.  51st IEEE  Conf. Dec. Contr.}, Hawaii, USA, 2012.


\bibitem{hespanha2008}
J.~P. Hespanha, D.~Liberzon,  A.~R. Teel, 
``Lyapunov conditions for input-to-state stability of impulsive systems,"
 \emph{Automatica}, vol. 44,  pp. 2735-2744, 2008.



\bibitem{higham2000}
D.~J. Higham, ``Mean-square and asymptotic stability of the stochastic theta method,"
  \emph{SIAM J. Numer. Anal.}, vol. 38,  pp. 753-769, 2000.


\bibitem{huang2009}
L.~Huang,  X.~Mao, ``On input-to-state stability of stochastic retarded systems with Markovian switching,"
  \emph{IEEE Trans.  Automat. Contr.}, vol. 54, pp. 1898-1902, 2009.

\bibitem{huang2009a}
L.~Huang, X.~Mao, ``Robust delayed-state-feedback stabilization of uncertain stochastic systems,"
  \emph{Automatica}, vol.45,  pp.1332-1339, 2009.

\bibitem{huang2010PhD}
L.~Huang, ``Stability and stabilisation of stochastic delay systems,"
 University of Strathclyde, Glasgow, UK, PhD thesis, 2010.


\bibitem{huang2011}
L.~Huang, X.~Mao, ``Stability of singular stochastic systems with Markovian switching,"
  \emph{IEEE Trans.  Automat. Contr.}, vol. 56, pp. 424-429, 2011.

\bibitem{huang2012}
L.~Huang, H.~Hjalmarsson, ``Recursive estimators with Markovian jumps,"
 \emph{Syst.   Contr. Lett.}, vol. 61,  pp. 405-412, 2012.


%
%



\bibitem{huang2015armax}
L.~Huang, H.~Hjalmarsson,   ``A multi-time-scale generalization of recursive identification algorithm for ARMAX systems,"
 \emph{IEEE Trans.  Automat. Contr.}, vol. 60, pp. 2242-2247, 2015.

\bibitem{huang2015}
L.~Huang, H.~Hjalmarsson,  H.~Koeppl, ``Almost sure stability and stabilization of discrete-time stochastic systems,"
  \emph{Syst. Contr. Lett.}, vol. 82,  pp. 26-32, 2015. 

\bibitem{huang2016}
L.~Huang, L.~Pauleve, C.~Zechner, M.~Unger, A.~S. Hansen, H.~Koeppl, ``Reconstructing dynamic molecular states from single-cell time series," \emph{J. R. Soc. Interface}, vol. 13: 20160533, 2016. 

\bibitem{huang202x}
L.~Huang, S.~Xu, ``Impulsive stabilization of  systems with control delay,"
  \emph{IEEE Trans.  Automat. Contr.}, accepted. 


\bibitem{huang_partI}
L.~Huang, ``Stability of cyber-physical systems of numerical methods for stochastic differential equations: integrating the cyber and the physical of stochastic systems,'' resubmitted. 


\bibitem{khalil2002}
H.~K. Khalil,   \emph{Nonlinear systems (3rd edition)},
  New Jersey, USA: Prentice Hall, 2002.

\bibitem{khas2012}
R.~Khasminskii,   \emph{Stochastic stability of differential equations (2nd ed.)},
  Berlin, Gernany: Springer-Verlag, 2012. 


\bibitem{krstic2008}
M. Krstic, A. Smyshlyaev,  Boundary control of PDEs: a course on backstepping designs. Philadelphia,  USA: Soc. Ind. Appl. Math., 2008. 


\bibitem{lee2010}
E.~A. Lee, ``CPS Foundations," 
in \emph{Proc. of 47th IEEE/ACM Design Automation Conf.}, Chicago, IL, USA, 2010.

 



\bibitem{liu1988}
Y.~Liu, Z.~Song,  \emph{Theory and application of large-scale dynamic systems (vol. 1): decomposition, stability and structure},  Guangzhou, China: South China Univ.  Technol. Press, 1988 (in Chinese).


\bibitem{liu1992}
Y.~Liu, Z.~Feng,  \emph{Theory and application of large-scale dynamic systems (vol. 4): stochastic stability and control},  Guangzhou, China: South China Univ.  Technol. Press, 1992 (in Chinese).




\bibitem{mao2007book}
X.~Mao,  \emph{Stochastic differential equations and applications (2nd ed.)}, Chichester, UK: Horwood Publishing, 2007. 



\bibitem{mao2013}
X.~Mao, ``Stabilization of continuous-time hybrid stochastic differential
equations by discrete-time feedback control,"
  \emph{Automatica}, vol.49,  pp.3677-3681, 2013.


\bibitem{mao2014}
X.~Mao, W.~Liu, L.~Hu, Q.~Luo and J.~Lu, ``Stabilization of hybrid stochastic differential equations by feedback
control based on discrete-time state observations,"
   \emph{Syst. Contr. Lett.}, vol.73,  pp.88-95, 2014.



\bibitem{maxwell1868}
J.~C. Maxwell,  ``On governors,"  \emph{Proc.  Roy.  Soc. London}, vol. 16,
pp. 270-283, 1868.



\bibitem{nagh2008}
 P.~Naghshtabrizi,  J.~P. Hespanha, A.~R. Teel, ``Exponential stability of impulsive systems with application to uncertain
sampled-data systems,"
  \emph{Syst. Contr. Lett.}, vol.57,  pp.378-385, 2008.  

\bibitem{nesic2001}
 D.~Nesic, A.~R. Teel, ``Sampled-data control of nonlinear systems: an overview of recent results," in: S.~R. Moheimani (eds) Perspectives in robust control. Lecture Notes in Control and Information Sciences, vol 268.  London, UK: Springer, 2001.



\bibitem{nesic2004}
D.~Nesic, A.~R. Teel,  ``A framework for stabilization of nonlinear sampled-data systems based on their approximate discrete-time models,"
  \emph{IEEE Trans.  Automat. Contr.}, vol. 49, pp. 1103-1034, 2004.

\bibitem{nesic2006}
 D.~Nesic, A.~R. Teel, ``Stabilization of sampled-data nonlinear systems via backstepping on their
Euler approximate model,"
  \emph{Automatica}, vol.42,  pp.1801-1808, 2006.

\bibitem{nghiem2012}
T.~Nghiem, G.~J. Pappas, R.~Alur, A.~Girard,  ``Time-triggered implementation of dynamic controllers," \emph{ACM Trans. Embed. Comput. Syst.}, vol.11,  no.S2, art.58, 2012.


\bibitem{nikoukhah1998}
R.~Nikoukhah, ``A new methodology for observer design and implementation,"
  \emph{IEEE Trans.  Automat. Contr.}, vol. 43, pp. 229-234, 1998.

\bibitem{oishil2010}
Y.~Oishi, H.~Fujioka, ``Stability and stabilization of aperiodic sampled-data control systems using
robust linear matrix inequalities,"
  \emph{Automatica}, vol.44,  pp.1327-1333, 2010.

\bibitem{qian2002}
C.~Qian,  W.~Lin, ``Output feedback control of a class of
nonlinear systems: a nonseperation principle paradigm," 
  \emph{IEEE Trans.  Automat. Contr.}, vol.47, pp.1710-1715, 2002.

\bibitem{rios2016}
H.~Rios, L.~Hetal, D.~Efimov,    ``Observer-based control for linear sampled-data systems: an impulsive system approach," 
in \emph{Proc.  55th IEEE  Conf. Dec. Contr.}, Las Vegas, USA, 2016.


\bibitem{samoilenko1995}
A.~M. Samoilenko,  N.~ A. Perestyuk,  
 \emph{Impulsive Differential Equations},
  Singapore: Word Scientific, 1995.



\bibitem{scotton2013}
F.~Scotton, L.~Huang, S.~A. Ahmadi, B.~Wahlberg,  ``Physics-based Modeling and Identification for HVAC Systems," 
in \emph{Proc.  Europ. Contr. Conf.}, Zurich, Switzerland, 2013.


\bibitem{seuret2012}
A.~Seuret,    ``A novel stability analysis of linear systems under asynchronous samplings," 
 \emph{Automatica}, vol. 48, pp. 177-103082, 2012.

\bibitem{stein2003}
G.~Stein, ``Respect the unstable," \emph{ IEEE Contr. Syst. Mag.}, vol. 23, iss. 4, pp. 12–25, 2003.

\bibitem{tanwani2016}
A.~Tanwani, C.~Prieur, M.~Fiacchini, ``Observer-based feedback stabilization of linear systems with event-triggered sampling and dynamic quantization,"
  \emph{Syst.  Contr. Lett.}, vol. 94,  pp. 46-56, 2016. 



\bibitem{teel2014}
A.~R. Teel, A.~Subbaraman,  A.~Sferlazza,  
``Stability analysis for stochastic hybrid systems: a survey,"
 \emph{Automatica}, vol.50,  pp. 2435-2456, 2014.





\bibitem{tsinias1991}
J.~Tsinias, ``A theorem on global stabilization of nonlinear systems by linear feedback,"
   \emph{Syst. Contr. Lett.}, vol.17,  pp.357-362, 1991.


\bibitem{wiener1961}
N.~Wiener,   \emph{Cybernetics or control and communication in the animal and the machine (2nd ed.)}, 
 Massachusetts, US: The M.I.T. Press, 1961.



\bibitem{yang2001}
T.~Yang,   \emph{Impulsive control theory},
  Berlin, Gernany: Springer-Verlag, 2001.





\bibitem{you2015}
S.~You, W.~Liu, J.~Lu, X.~Mao and Q.~Qiu,  
``Stabilization of hybrid systems by feedback control based on discrete-time state observations,"
 \emph{SIAM J. Control Optim.}, vol. 53, pp. 905-925, 2015.

\bibitem{zheng2002}
D.~Zheng,   \emph{Linear system theory (2nd ed.)},  Beijing, China: Tsinghua Univ. Press, 2002 (in Chinese).


\end{thebibliography}
\end{document}